\documentclass[acmsmall,screen]{acmart}

\usepackage{amsmath,amsfonts,amsthm,stmaryrd}
\usepackage{xspace}
\usepackage{thmtools}
\usepackage{mathpartir}
\usepackage{natbib}
\usepackage[shortlabels,inline]{enumitem}
\theoremstyle{definition}\newtheorem{definition}{Definition}[section]
\newtheorem{theorem}[definition]{Theorem}
\theoremstyle{theorem}
\theoremstyle{theorem}\newtheorem{lemma}[definition]{Lemma}
\theoremstyle{theorem}
\theoremstyle{theorem}\newtheorem{corollary}[definition]{Corollary}
\theoremstyle{example}
\theoremstyle{remark}
\usepackage{hyperref}
\usepackage{tikz-cd}
\usepackage{bussproofs}
\EnableBpAbbreviations 
\usepackage{color}
\usepackage{algorithmicx}
\usepackage{algpseudocode}
\usepackage{algorithm}
\usepackage{epigraph}
\usepackage{suffix}

\newcommand\nin{\not\in}
\newcommand\sym[1]{\mathsf{#1}}
\newcommand\pair[2]{\left\langle#1,#2\right\rangle}
\newcommand\set[1]{\left\{#1\right\}}
\newcommand{\setcom}[2]{\set{#1\left\vert\vphantom{#1}\,#2\right.}}

\newcommand\allocate[1]{\sym{alloc}\left(#1\right)}

\newcommand\dom[1]{\sym{dom}\left(#1\right)}
\newcommand\ol[1]{\overline{#1}}

\newcommand\tickle{\sqsubseteq_{\tick}}

\newcommand\heaple{\sqsubseteq}

\newcommand\labelfree{\sym{location}\text{-}\sym{free}}

\newcommand{\interm}{\sym{in}}
\newcommand{\caseterm}[5]{\sym{case}\, #1 \, \sym{of}\, \interm_1\,#2 . #3 ; \interm_2\,#4 . #5}

\newcommand\unit{\langle\rangle}
\newcommand\zero{0}
\newcommand\suc{\sym{suc}}

\newcommand{\true}{\sym{tt}}

\newcommand\fix{\sym{fix}}
\def\nothing{}
\newcommand\tick[1][\nothing]{
  \ifthenelse{\equal{#1}{\nothing}}{\checkmark}{
  \checkmark_{\hspace{-4pt}#1\hspace{1pt}}}}
\newcommand\lock{\sharp}
\newcommand\now{\sym{now}}
\newcommand\wait{\sym{wait}}
\newcommand\stay{\sym{stay}}
\newcommand\switch{\sym{switch}}
\WithSuffix\newcommand\stay*{\sym{stay}'}
\WithSuffix\newcommand\switch*{\sym{switch}'}
\newcommand\adv{\sym{adv}}
\newcommand\delay{\sym{delay}}
\newcommand\rbox{\sym{box}}
\newcommand\promote{\sym{promote}}
\newcommand\progress{\sym{progress}}
\newcommand\unbox[1][]{\sym{unbox}^{#1}}
\newcommand\out{\sym{out}}
\newcommand\into{\sym{into}}
\newcommand{\head}{\sym{head}}
\newcommand{\tail}{\sym{tail}}

\newcommand{\gc}[1]{\sym{gc}\left( #1 \right)}

\newcommand{\recN}[5]{\sym{rec}_{\Nat}(#1,#2\,#3.#4,#5)}
\newcommand\recU[7]{\sym{rec}_{\Until}(#1.#2,#3\,#4\,#5.#6,#7)}
\WithSuffix\newcommand\recU*{\sym{rec}_{\Until}}
\newcommand{\delayapp}{\ensuremath{\circledast}}
\newcommand{\laterapp}{\ensuremath{\rhd\hspace{-10pt}*\hspace{5pt}}}
\newcommand{\progressapp}{\ensuremath{\odot}}
\newcommand{\boxapp}{\ensuremath{\boxast}}

\newcommand{\event}[1]{\sym{Ev}(#1)}
\WithSuffix\newcommand\event*{\sym{Ev}}

\newcommand\Delay{\ensuremath{{\bigcirc}}}
\newcommand\Later{\ensuremath{{\rhd}}}

\newcommand\stable{\sym{stable}}
\newcommand\limit{\sym{limit}}
\newcommand\Str[1]{\sym{Str}(#1)}
\newcommand\Event[1]{\sym{Ev}(#1)}
\newcommand\Nat{\sym{Nat}}
\newcommand\Fair[2]{\sym{Fair(#1,#2)}}
\WithSuffix\newcommand\Fair*[2]{\sym{Fair'(#1,#2)}}
\newcommand\TNat{\Nat_\Delay}
\newcommand\Unit{\sym{1}}
\newcommand\Fix{\sym{Fix}}

\renewcommand\Until{\mathrel{\mathcal U}}
\newcommand{\Finally}{\Diamond}

\newcommand\nats{{\mathbb N}}

\newcommand{\iso}{\cong}

\newcommand{\vinterpN}[2]{\mathcal{V}\llbracket #1 \rrbracket(#2)}
\newcommand{\tinterpN}[2]{\mathcal{T}\llbracket #1 \rrbracket(#2)}
\newcommand{\cinterpN}[2]{\mathcal{C}\llbracket #1 \rrbracket(#2)}

\newcommand\dinterp[1]{\llbracket #1 \rrbracket}

\newcommand\wfcxt[2][]{#2 \vdash^{#1} }

\newcommand\hastype[4][]{#2 \vdash^{#1} #3 : #4}

\newcommand\tokenfree[1]{\sym{token}\text{-}\sym{free}(#1)}
\newcommand\tickfree[1]{\sym{tick}\text{-}\sym{free}(#1)}

\newcommand\lockfree[1]{\sym{lock}\text{-}\sym{free}(#1)}

\renewcommand\state[2]{\left\langle #1;#2 \right\rangle}
\newcommand\stateS[3]{\left\langle #1;#2;#3 \right\rangle}
\newcommand\stateF[4]{\left\langle #1;#2;#3;#4 \right\rangle}
\newcommand\heval[4]{ \state{#1}{#2} \Downarrow \state{#3}{#4}}

\newcommand{\defeq}{\mathbin{\overset{\textsf{def}}{=}}}

\newcommand\locs{\sym{Loc}}

\newcommand\tsize[1]{\left| #1\right|}

\WithSuffix\newcommand\state*[3]{\left\langle #1;#2;#3 \right\rangle}
\newcommand\forward[2]{\stackrel{#2\hphantom{#1}}{\Longrightarrow_{#1}}}
\WithSuffix\newcommand\forward*[3]{\stackrel{#2/#3\hphantom{#1}}{\Longrightarrow_{#1}}}
\newcommand\forwards[1]{\forward{\sym{Str}}{#1}}
\WithSuffix\newcommand\forwards*[2]{\forward*{\sym{Str}}{#1}{#2}}
\newcommand\forwardu[1]{\forward{\Until}{#1}}
\WithSuffix\newcommand\forwardu*[2]{\forward*{\Until}{#1}{#2}}
\newcommand\forwardf[1]{\forward{\sym{F}}{#1}}
\WithSuffix\newcommand\forwardf*[2]{\forward*{\sym{F}}{#1}{#2}}

\newcommand{\supported}{\sym{supported}}
\newcommand{\support}{\bowtie}
\newcommand{\nullstore}{\bullet}

\input{examples.tex}

\usepackage{booktabs}   %% For formal tables:
\usepackage{subcaption} %% For complex figures with subfigures/subcaptions
\usepackage{fancyvrb}
\begin{document}

%% Title information
\title{Diamonds are not forever}
%\title{To Infinity and Beyond}

%\titlenote{with title note}             %% \titlenote is optional;
                                        %% can be repeated if necessary;
                                        %% contents suppressed with 'anonymous'
\subtitle{Liveness in reactive programming with guarded recursion}                     %% \subtitle is optional
%\subtitle{Fair reactive programming with guarded recursion}                     %% \subtitle is optional
%\subtitlenote{with subtitle note}       %% \subtitlenote is optional;
                                        %% can be repeated if necessary;
                                        %% contents suppressed with 'anonymous'

%% Author information
%% Contents and number of authors suppressed with 'anonymous'.
%% Each author should be introduced by \author, followed by
%% \authornote (optional), \orcid (optional), \affiliation, and
%% \email.
%% An author may have multiple affiliations and/or emails; repeat the
%% appropriate command.
%% Many elements are not rendered, but should be provided for metadata
%% extraction tools.

%% Author with single affiliation.

\author{Patrick Bahr}
%\authornote{with author1 note}          %% \authornote is optional;
                                        %% can be repeated if necessary
%\orcid{nnnn-nnnn-nnnn-nnnn}             %% \orcid is optional
\affiliation{
 % \position{Position1}
 % \department{Department1}              %% \department is recommended
  \institution{IT University of Copenhagen}            %% \institution is required
%  \streetaddress{Street1 Address1}
%  \city{City1}
%  \state{State1}
%  \postcode{Post-Code1}
  \country{Denmark}                    %% \country is recommended
}
\email{paba@itu.dk}          %% \email is recommended

\author{Christian Uldal Graulund}
%\authornote{with author2 note}          %% \authornote is optional;
                                        %% can be repeated if necessary
%\orcid{nnnn-nnnn-nnnn-nnnn}             %% \orcid is optional
\affiliation{
%  \position{Position2a}
%  \department{Department2a}             %% \department is recommended
  \institution{IT University of Copenhagen}           %% \institution is required
%  \streetaddress{Street2a Address2a}
%  \city{City2a}
%  \state{State2a}
%  \postcode{Post-Code2a}
  \country{Denmark}                   %% \country is recommended
}
\email{cgra@itu.dk}         %% \email is recommended

\author{Rasmus Ejlers M\o{}gelberg}
%\authornote{with author2 note}          %% \authornote is optional;
                                        %% can be repeated if necessary
%\orcid{nnnn-nnnn-nnnn-nnnn}             %% \orcid is optional
\affiliation{
%  \position{Position2a}
%  \department{Department2a}             %% \department is recommended
  \institution{IT University of Copenhagen}           %% \institution is required
%  \streetaddress{Street2a Address2a}
%  \city{City2a}
%  \state{State2a}
%  \postcode{Post-Code2a}
  \country{Denmark}                   %% \country is recommended
}
\email{mogel@itu.dk}         %% \email is recommended

%% Abstract
%% Note: \begin{abstract}...\end{abstract} environment must come
%% before \maketitle command
\begin{abstract}
When designing languages for functional reactive programming (FRP) the
main challenge is to provide the user with a simple, flexible interface for writing 
programs on a high level of abstraction while ensuring that all programs can be 
implemented efficiently in a low-level language. To meet this challenge, a new 
family of modal FRP languages has been proposed, in which variants of Nakano's
guarded fixed point operator are used for writing recursive programs guaranteeing 
properties such as causality and productivity. As an apparent extension to this it 
has also been suggested to use Linear Temporal Logic (LTL) as a language for 
reactive programming through the Curry-Howard isomorphism, allowing 
properties such as termination, liveness and fairness to be encoded in types. 
However, these two ideas are in conflict with each other, since the 
fixed point operator introduces non-termination into the inductive
types that are supposed to provide termination guarantees.

In this paper we show that by regarding the modal time step operator of LTL
a submodality of the one used for guarded recursion (rather than equating them), 
one can obtain a modal type
system capable of expressing liveness properties while retaining the power of the guarded fixed 
point operator. We introduce the language Lively RaTT, a modal FRP language
with a guarded fixed point operator and an `until' type constructor as in LTL, and 
show how to program with events and fair streams. Using a step-indexed Kripke
logical relation we prove operational properties of Lively RaTT including productivity 
and causality as well as the termination and liveness properties expected of types 
from LTL. Finally, we prove that the type system of Lively RaTT guarantees the 
absence of implicit space leaks. 
\end{abstract}

%
%%% 2012 ACM Computing Classification System (CSS) concepts
%%% Generate at 'http://dl.acm.org/ccs/ccs.cfm'.
%\begin{CCSXML}
%<ccs2012>
%<concept>
%<concept_id>10011007.10011006.10011008.10011009.10011012</concept_id>
%<concept_desc>Software and its engineering~Functional languages</concept_desc>
%<concept_significance>500</concept_significance>
%</concept>
%<concept>
%<concept_id>10011007.10011006.10011008.10011009.10011016</concept_id>
%<concept_desc>Software and its engineering~Data flow languages</concept_desc>
%<concept_significance>500</concept_significance>
%</concept>
%<concept>
%<concept_id>10011007.10011006.10011008.10011024.10011033</concept_id>
%<concept_desc>Software and its engineering~Recursion</concept_desc>
%<concept_significance>300</concept_significance>
%</concept>
%<concept>
%<concept_id>10003752.10010124.10010131.10010134</concept_id>
%<concept_desc>Theory of computation~Operational semantics</concept_desc>
%<concept_significance>300</concept_significance>
%</concept>
%</ccs2012>
%\end{CCSXML}

%% End of generated code

%% Keywords
%% comma separated list
%% \keywords are mandatory in final camera-ready submission

%% \maketitle
%% Note: \maketitle command must come after title commands, author
%% commands, abstract environment, Computing Classification System
%% environment and commands, and keywords command.

\maketitle
%
%\begin{epigraphs}
%  \qitem{\textit{%
%      If you don’t know me by now,\\
%      you will never never never know me.}}%
%  {Simply Red}
%  \qitem{\textit{%
%      You should know by now.\\
%      You should know by now.\\
%      You should know by now.\\
%      You should know by now.}}%
%  {Ratt}
%\end{epigraphs}

\section{Introduction}
\label{sec:introduction}

Reactive programs such as servers and control software in cars,
aircrafts and robots are traditionally written in imperative languages
using a wide range of complex features including call-backs and
shared state. For this reason, they are notoriously error-prone and
hard to reason about. This is unfortunate, since much of the most
critical software currently in use is reactive. The goal of functional
reactive programming (FRP) is to provide the programmer with tools for
writing reactive programs on a high level of abstraction in the
functional paradigm. In doing so, FRP extends the known benefits of
functional programming also to reactive programming, in particular
modularity and equational reasoning for programs. The challenge for
achieving this goal is to ensure that all programs can be implemented
efficiently in a low-level language.

From the outset, the central idea of FRP~\citep{FRAN} was that
reactive programming simply is programming with signals and events. While elegant, this idea
immediately leads to the question of what the interface for signals
and events should be. A naive approach would be to model signals as
streams in the sense of coinductive solutions to
$\Str{A} \iso A \times \Str A$, but this allows the programmer to write \emph{non-causal}
programs, i.e., programs where the present output depends on future input. Arrowised
FRP~\citep{nilsson2002}, as implemented in the Yampa library for
Haskell, solves this problem by taking signal functions as primitive
rather than signals themselves.  However, this approach forfeits some
of the simplicity of the original FRP model and reduces its
expressivity as it rules out useful types such as signals of signals.

More recently, a number of 
authors~\citep{jeffrey2014,krishnaswami2011ultrametric,krishnaswami2012higher,krishnaswami13frp,
jeltsch2013temporal,bahr2019simply} have suggested a modal approach to FRP in which causality
is ensured through the introduction of a notion of time in the form of a modal operator. In this approach, 
an element of the modal type $\Later A$ should be thought of as data of type $A$ arriving in the next time step. 
Signals should be modelled as a type of streams satisfying the type isomorphism $\Str{A} \iso A \times \Later \Str A$ capturing the
idea that each pair of elements of a stream is separated by a time step. Events carrying data of type $A$ 
can be represented by a type satisfying $\Event A \iso A + \Later \Event A$, stating
that an event can either occur now, or at some point in the future. Types such as $\Str A$ and $\Event A$ 
satisfying type equations in which the recursion variable is guarded by a $\Later$ are referred to as 
\emph{guarded recursive types}. Combining this with guarded recursion
\citep{nakano2000} in the form of a fixed point operator of type $(\Later A \to A) \to A$ gives a powerful type system for 
reactive programming guaranteeing not only causality, but also productivity, i.e. the property that for a closed
stream, each of its elements can always be computed in finite time. In some 
systems~\citep{krishnaswami13frp,bahr2019simply,krishnaswami2012higher}
the modal types have also been used to guarantee the lack of implicit space leaks, 
i.e., the problem of programs holding on to 
memory while continually allocating until they run out of space. These leaks have previously
been a major problem in FRP. 

\citet{jeffrey2012} suggested taking this idea further using Linear Temporal Logic (LTL) \citep{LTL} as a
type system for FRP through the Curry-Howard isomorphism, a connection discovered independently by
\citet{Jeltsch2012}. This idea is not only conceptually 
appealing, but could also extend the expressivity of the type system considerably and have practical consequences.
Indeed, LTL has a step modality $\Delay$ similar to $\Later$ used to express that a formula should
be true one time step from now. It also has an operation $\Box$ expressing global truth, i.e.,
formulas that hold now and at any time in the future. This operation has been used by 
\citet{krishnaswami13frp} to express 
time-independent data that can be safely kept across time steps without causing space leaks. In this paper
we are particularly interested in the \emph{until} operator $\phi \Until \psi$ of LTL, which expresses that $\phi$
holds now and for some more steps, after which $\psi$ becomes
true. Using this operator, we can encode the 
\emph{finally} operator $\Finally \phi$ as $\true \Until \phi$ stating that $\phi$ will eventually become true. 
In programming terms, the until operator is the inductive type given by constructors
\begin{align*}
 \now & :  A \to A \Until B & \wait & : A \to \Delay (A \Until B) \to A \Until B \, .
\end{align*}
and the fact that it is inductive should imply a termination property similarly to that of LTL: 
Elements of type $A \Until B$ will eventually produce a $B$ after at most finitely many $A$s, 
and similarly for elements of type $\Finally B$. In programming, this can be used to express the property
that a program will eventually produce an output, e.g., by timeout, or one can give a type of fair schedulers
\citep{cave14fair}, see \autoref{sec:fair:streams} for details. 

The goal of this paper is to define a language combining the expressive type system of LTL 
with the power of the guarded recursive fixed point combinator. Unfortunately, equating 
$\Delay$ and $\Later$ in such a system breaks the termination guarantee of the $\Until$ type. For example,
if $a : A$, the fixed point of $\wait \, a : \Delay(A \Until B) \to A
\Until B$ will never produce a $B$. This is an example of a well-known
phenomenon: The guarded fixed point combinator implies uniqueness of solutions 
to guarded recursive type
equations like $X \iso B + A \times \Later X$, and so inductive and coinductive solutions 
coincide. 
In fact, the solutions behave more like coinductive types than inductive types and can even be
used to encode coinductive types \citep{atkey2013productive} in some settings. 

This observation led \citet{cave14fair} to suggest removing the guarded recursive fixed point operator
from FRP in order to distinguish between inductive and coinductive guarded
types. This has the unfortunate effect of
losing the power and elegance of the guarded fixed point operator for programming with
coinductive types, which ought to be safe. 
Indeed it is well known that programming directly with coiteration 
is cumbersome and so most programming languages
allow the programmer to construct elements of coinductive types using recursion. To
guarantee productivity, one must use either the (non-modular) syntactic checks
used in most proof assistants today, or sized types~\citep{Abel:Wellfounded,Abel:NBE:sized:types,
Sacchini13,HughesPS96}. Given that the modal operator is in the 
language, guarded recursion is the most obvious solution to guaranteeing productivity.

\subsection{Overview of results}

In this paper we show that by considering $\Delay$ a submodality of $\Later$, rather than equating them,
we can use the guarded fixed point operator while retaining the termination guarantees of $\Until$. 
Using $\Later$, the type $\Event A$ of possibly occurring events of type $A$ can be encoded as 
the unique solution to $\Event A \iso A + \Later \Event A$. Using $\Delay$, the type $\Finally A$ 
of events of type $A$ that must occur can be encoded as above. We will often refer to these as
the types of possibly non-terminating and terminating events, respectively. The inclusion from $\Delay$
into $\Later$ can be used to type an inclusion of $\Finally A$ into $\Event A$. 
%
%The inclusion of modalities allows us to type an inclusion of the type $\Finally A$ of events 
%of type $A$ guaranteed to occur into the type $\Event A$ of events that may or may not occur,
%using that $\Finally A$ is the inductive solution to $\Finally A \iso A + \Delay\Finally A$ and the fact that 
%$\Event A \iso A + \Later \Event A$. 
The lack of an
inclusion from $\Later$ to $\Delay$ means that there is no inclusion $\Later \Finally A \to \Finally A$
to take a fixed point of to construct a diverging element of $\Finally A$. 

To make these ideas concrete we define the language \emph{Lively RaTT} (\autoref{sec:live:ratt}) as 
an extension of the language Simply RaTT~\citep{bahr2019simply}. Simply RaTT is an FRP language
with modal operators $\Later$ and $\Box$ as described above, as well as guarded recursive types
and guarded fixed points. It uses a Fitch-style approach~\citep{Fitch:Symbolic,Clouston:fitch-2018,
clouston2018modal} to programming with modal types, which means that the typing rules for 
introduction and elimination for modal types add and remove tokens from a context.
%
%contexts can contain tokens like $\tick$ and $\lock$ each associated with modal operators.
%In Simply RaTT the token $\tick$ is associated with $\Later$ and $\lock$ with $\Box$. Intuitively a $\tick$ in
%a context represents a step of time: Variables to the left of the $\tick$ are available one time step 
%before those to the right. The token $\lock$ on the other hand, represents the separation between 
%initialisation time and dynamic time, i.e., the time where the program is run in a step-by-step fashion.
%Introduction rules for modal types abstract 
%these tokens, and elimination is by application to a fresh token. 
This gives a direct style for programming
with modalities, avoiding let-expressions as traditionally used for elimination. 
Lively RaTT has tokens $\tick[\Later]$ and $\tick[\Delay]$ for $\Later$ and $\Delay$, respectively, and 
the inclusion of $\Delay$ into $\Later$ is defined by allowing $\tick[\Later]$ to eliminate also $\Delay$.
We think of $\tick[\Delay]$ and $\tick[\Later]$ as a separation in time in judgements: Variables to the left of $\tick[\Delay]$ or $\tick[\Later]$ are 
available one time step before those to the right. The token $\tick[\Later]$ is a stronger time step, 
allowing also recursive definitions to be unfolded. 
We illustrate the expressivity of Lively RaTT by showing how to program with events and fair streams
in \autoref{sec:examples}. 

We define two kinds of operational semantics for Lively RaTT (\autoref{sec:op:sem}): An evaluation semantics 
reducing terms to values at each time instant, and a step semantics capturing the dynamic behaviour of reactive 
programs over time. The latter is defined for streams, 
$\Until$-types, and fair streams only.
We prove causality and productivity of streams, and we prove the termination property
for $\Until$-types, i.e., that any term
of type $A \Until B$ eventually produces a $B$, also in a context
of a stream of external inputs. Using this, we prove that any term of the fair scheduler type can be 
unwound to a fair interleaving of streams, again also in a context of external input. 

Finally we show that the type system of Lively RaTT guarantees the lack of implicit space leaks.
Our results on this extend those proved for Simply RaTT by 
\citet{bahr2019simply} which in turn
were based on a technique developed by~\citep{krishnaswami13frp}. More precisely, our operational
semantics stores input as well as delayed computations in a heap, and we show that it is safe to 
garbage collect the elements in the heap after two evaluation steps. %\rasmus{forward reference}

These results are proved (\autoref{sec:metatheory}) using an interpretation of types as sets of 
values indexed by four parameters, including an ordinal $\beta$. For finite $\beta$, this index should be thought of
as a form of step-indexing: The interpretation of $A$ at $\beta$ in this case describes the behaviour of
terms up to the first $\beta$ evaluation steps. In our model, however, $\beta$ runs all the way to $\omega\cdot 2$.
The interpretation at higher $\beta$, in particular the limit ordinal $\omega$ describes
global behaviour of programs. 

The distinction between $\Later$ and $\Delay$ can be seen in the model. At successor ordinals $\beta +1$,
the interpretation of $\Later A$ and $\Delay A$ are both defined in terms of the interpretation of $A$ at $\beta$
in a step-indexed fashion \citep{ToT}, but at limit ordinals $\beta$, the interpretation of $\Later A$ is the intersection 
of the interpretations at $\beta'<\beta$, whereas $\Delay A$ is interpreted using the interpretation of $A$ at 
$\beta$. This interpretation of $\Later A$ is needed to interpret fixed points, and the interpretation of 
$\Delay A$ ensures that the interpretation of $A \Until B$ behaves globally as an inductive type.  

The paper ends with an overview of related work
(\autoref{sec:related:work}) and conclusions, perspectives and future
work (\autoref{sec:conclusion}). The accompanying technical appendix
contains the full proofs of our results.

\section{Lively RaTT}
\label{sec:live:ratt}

Lively RaTT is an extension of Simply RaTT~\citep{bahr2019simply}, a
Fitch-style modal language for reactive programming. 
%The type system of Simply RaTT
%guarantees the lack of implicit space leaks, i.e., the problem of programs holding on to 
%memory while continually allocating until they run out of space. 
%Although we do not extend the results on space leaks proved
%for Simply RaTT to Lively RaTT, we do maintain the restrictions on the
%language known to be necessary for these. 
This section gives an overview
of the language, referring to \autoref{fig:typing} for an overview of the typing rules.

\begin{figure}[tbp]
\begin{center}
%  \textbf{Types} 
\begin{align*}
\text{\textbf{Types}} &&  A, B &:: = \alpha \mid \Unit \mid \Nat \mid A \times B \mid A + B \mid A \to B \mid \Box A \mid \Delay A \mid \Later A 
  \mid \Fix\;\alpha.A \mid A \Until B \\
%\end{align*}
%  \textbf{Stable types} 
%\begin{align*}
\text{\textbf{Stable types} } && S, S' &:: = \Unit \mid \Nat \mid \Box A \mid S \times S' \mid S + S' \\
%\end{align*}
%%\begin{mathpar}
%%  \inferrule*
%%  {~}
%%  {\Unit \; \stable}
%%  \and
%%  \inferrule*
%%  {~}
%%  {\Nat \; \stable}
%%  \and
%%  \inferrule*
%%  {~}
%%  {\Box A \; \stable}
%%  \and
%%  \inferrule*
%%  {A \; \stable \\ B \; \stable}
%%  {A \times B \; \stable}
%%  \and
%%  \inferrule*
%%  {A \; \stable \\ B \; \stable}
%%  {A + B \; \stable}
%%\end{mathpar}
%
%  \textbf{Limit types}
%  \begin{align*}
 \text{\textbf{Limit types}} && L, L' &:: = \alpha \mid \Unit \mid \Nat \mid \Later A \mid \Delay L \mid \Box L \mid L \times L' \mid L + L' \mid A \to L \mid \Fix\;\alpha.L
\end{align*}
%
%\begin{mathpar}
%  \inferrule*
%  {~}
%  {\alpha \; \limit}
%  \and
%  \inferrule*
%  {~}
%  {\Unit \; \limit}
%  \and
%  \inferrule*
%  {~}
%  {\Later A \; \limit}
%  \and
%  \inferrule*
%  {A \; \limit}
%  {\Box A \; \limit}
%  \and
%  \inferrule*
%  {A \; \limit}
%  {\Delay A \; \limit}
%  \and 
%  \inferrule*
%  {A \; \limit \\ B \; \limit}
%  {A \times B \; \limit}
%  \and
%  \inferrule*
%  {A \; \limit \\ B \; \limit}
%  {A + B \; \limit}
%  \and
%  \inferrule*
%  {B \; \limit}
%  {A \to B \; \limit}
%  \and
%  \inferrule*
%  {A \; \limit}
%  {\Fix\;\alpha.A \; \limit}
%\end{mathpar}
\caption{Grammars for types, stable types and limit types. In typing rules, only closed types (no free $\alpha$) are considered.}
\label{fig:stable:and:limit}
\end{center}
\end{figure}

\begin{figure}[tbp]
\begin{center}
\begin{mathpar}
  \inferrule*
  {~}
  {\wfcxt{\cdot}}
  \and
  \inferrule*
  {\wfcxt{\Gamma} \quad x \nin \dom{\Gamma}}
  {\wfcxt{\Gamma,x:A}}
  \and
  \inferrule*
  {\wfcxt{\Gamma} \quad \lockfree{\Gamma}}
  {\wfcxt{\Gamma,\lock}}
  \and
  \inferrule*
  {\wfcxt{\Gamma} \quad 
  \lock \in \Gamma \quad m \in \{\Delay,\Later\} \quad \tickfree{\Gamma}}
  {\wfcxt{\Gamma,\tick[m]}}
\end{mathpar}
\caption{Well-formed contexts.}
\label{fig:wf:cxt}
\end{center}
\end{figure}

\begin{figure}
  \textbf{Simply typed $\lambda$-calculus:}
\begin{mathpar}
  \inferrule*
  {%\wfcxt{\Gamma,x:A,\Gamma'} \\ 
   \tokenfree{\Gamma'} \lor A \; \stable}
  {\hastype{\Gamma,x:A,\Gamma'}{x}{A}}
  \and
  \inferrule*{~}
  %{\wfcxt{\Gamma}}
  {\hastype{\Gamma}{\unit}{\Unit}}
  \and
  \inferrule*
  {\hastype{\Gamma,x:A}{t}{B} \\ \tickfree{\Gamma}}
  {\hastype{\Gamma}{\lambda x.t}{A \to B}}
  \and
  \inferrule*
  {\hastype{\Gamma}{t}{A \to B} \\
    \hastype{\Gamma}{t'}{A}}
  {\hastype{\Gamma}{t\,t'}{B}}
  \and
  \inferrule*
  {\hastype{\Gamma}{t}{A} \\
    \hastype{\Gamma}{t'}{B}}
  {\hastype{\Gamma}{\pair{t}{t'}}{A \times B}}
  \and
  \inferrule*
  {\hastype{\Gamma}{t}{A_1 \times A_2} \\ i \in \{1, 2\}}
  {\hastype{\Gamma}{\pi_i\,t}{A_i}}
  \and
  \inferrule*
  {\hastype{\Gamma}{t}{A_i} \\ i \in \{1, 2\}}
  {\hastype{\Gamma}{\interm_i\, t}{A_1 + A_2}}
  \and
  \inferrule*
  {\hastype{\Gamma,x: A_i}{t_i}{B}
    \\ \hastype{\Gamma}{t}{A_1 + A_2}
    \\ i \in \{1,2\}}
  {\hastype{\Gamma}{\caseterm {t}{x}{t_1}{x}{t_2}}{B}}
  \and
    \inferrule*
  {~}
  {\hastype{\Gamma}{\zero}{\Nat}}
  \and
  \inferrule*
  {\hastype{\Gamma}{t}{\Nat}}
  {\hastype{\Gamma}{\suc \, t}{\Nat}}
  \and
  \inferrule*
  {\hastype{\Gamma}{s}{A} \\ \hastype{\Gamma,x:\Nat,y:A}{t}{A}
  \\ \hastype{\Gamma}{n}{\Nat}}
  {\hastype{\Gamma}{\recN{s}{x}{y}{t}{n}}{A}}

\end{mathpar}
\noindent
%\textbf{2: Reactive terms:}
 \textbf{Modalities, $\Until$-types, guarded recursion:}
\begin{mathpar}
  \inferrule*
  {\hastype{\Gamma,\tick[m]}{t}{A}} 
%{{  \\ m =
%    \Delay \text{ or } \tickfreem[\Later]{\Gamma}} 
  {\hastype{\Gamma}{\delay\,t}{m\; A}}
  \and
  \inferrule*
  {\hastype{\Gamma}{t}{m\; A} \\ %\tokenfree{\Gamma'} \\ % \lock \in \Gamma \\ 
    m \leq m' \lor A \; \limit}
  {\hastype{\Gamma,\tick[m'],\Gamma'}{\adv\,t}{A}}
  \and
  \inferrule*
  {\hastype{\Gamma}{t}{\Box A} }
  {\hastype{\Gamma,\lock,\Gamma'}{\unbox\,t}{A}}
  \and
  \inferrule*
  {\hastype{\Gamma,\lock}{t}{A}}
  {\hastype{\Gamma}{\rbox\,t}{\Box A}}
%  \and
%  \inferrule*
%  {\hastype{\Gamma}{t}{A} \\ %\wfcxt{\Gamma,\tick[m],\Gamma'} \\ 
%   A \; \stable \\ \tokenfree{\Gamma'} }
%  {\hastype{\Gamma,\tick[m],\Gamma'}{\progress\,t}{A}}
%  \and
%  \inferrule*
%  {\hastype{\Gamma}{t}{A} \\ % \wfcxt{\Gamma,\lock,\Gamma'} \\ 
%  A \; \stable}
%  {\hastype{\Gamma,\lock,\Gamma'}{\promote\,t}{A}}
%\end{mathpar}
%\noindent
%%\textbf{3: Inductive terms:}
%\begin{mathpar}
  \and
  \inferrule*
  {\hastype{\Gamma}{t}{B}}
  {\hastype{\Gamma}{\now\,t}{A \Until B}}
 \and
  \inferrule*
  {\hastype{\Gamma}{s}{A} \\ \hastype{\Gamma}{t}{\Delay(A \Until B)}}
  {\hastype{\Gamma}{\wait\,s\,t}{A \Until B}}
  \and
  \inferrule*
  {\hastype{\Gamma,\lock,x:B}{s}{C}\\\hastype{\Gamma,\lock,x: A, y:
      \Delay (A \Until B), z : \Delay C}{t}{C}\\
    \hastype{\Gamma,\lock,\Gamma'}{u}{A \Until B}}
  {\hastype{\Gamma,\lock,\Gamma'}{\recU{x}{s}{x}{y}{z}{t}{u}}{C}}
%\end{mathpar}
%\textbf{4: Guarded recursive terms:}
%\begin{mathpar}
  \and
  \inferrule*
  {\hastype{\Gamma,x : \Box\Later A,\lock}{t}{A}}
  {\hastype{\Gamma}{\fix \; x.t}{\Box A}}
  \and
  \inferrule*
  {\hastype{\Gamma}{t}{\Fix \;\alpha.A}}
  {\hastype{\Gamma}{\out\,t}{A[\Later(\Fix \; \alpha.A)/\alpha]}}
  \and
    \inferrule*
  {\hastype{\Gamma}{t}{A[\Later(\Fix\,\alpha.A)/\alpha]}}
  {\hastype{\Gamma}{\into\,t}{\Fix\,\alpha.A}}
\end{mathpar}

  \caption{Typing rules. Here $m,m'$ ranges over the set $\{\Delay,\Later\}$ of time modalities ordered by $\Delay\leq \Later$.
  In all rules, all contexts are assumed well-formed.}
  \label{fig:typing}
\end{figure}

In the Fitch-style approach to modal types the 
introduction and elimination rules for these add and remove 
tokens from a context.
For example, the modality $\Delay$ expresses
delay of data by one time step and has introduction and elimination rules as follows
(ignoring $\Later$ for the moment).
\begin{mathpar}
  \inferrule*
   {\hastype{\Gamma,\tick[\Delay]}{t}{A}} 
  {\hastype{\Gamma}{\delay\,t}{\Delay A}}
  \and
  \inferrule*
  {\hastype{\Gamma}{t}{\Delay A}}% \\ \tokenfree{\Gamma'}}
  {\hastype{\Gamma,\tick[\Delay],\Gamma'}{\adv\,t}{A}}
\end{mathpar}
The token $\tick[\Delay]$ should be thought of as a separation by a single time step between the variables to the left of it 
and the rest of the judgement to the right. Thus the premise of the
introduction rule states that $t$ has type $A$ one time step after $\Gamma$, and so 
$\delay\,t$ has type $\Delay A$ at the time of $\Gamma$. Similarly, in the conclusion of the elimination rule,
one time step has passed since the premise, so at that time $t$ can be advanced to give an
element of type $A$. This gives a direct approach to programming with modalities, as opposed to
the more standard let-expressions. For example, delayed application of functions can be typed as
\begin{equation}
  \label{eq:delayApp}
  \lambda f . \lambda x . \delay((\adv \, f) (\adv\, x)) : \Delay(A \to B) \to \Delay A \to \Delay B
\end{equation}
%In the elimination rule, the premise $\tokenfree{\Gamma'}$ states that $\Gamma'$ does not 
%contain tokens like $\tick[\Delay]$. This prevents $\adv\,t$ from being transported further into the
%future, a common source of space leaks. Similarly, variables can not be introduced over 
%tokens. 

Unlike Simply RaTT, Lively RaTT has two modalities for time delays: $\Delay$ and $\Later$. 
Both correspond to a time step in the execution of reactive programs, but in addition, $\Later$ 
corresponds to a time step in the sense of guarded recursion. Consequently, the $\tick[\Later]$ token is 
stronger than $\tick[\Delay]$: Both can be used to advance time, but $\tick[\Later]$ can also
be used to unfold fixed points. We capture this extra strength in a reflexive ordering 
generated by $\Delay \leq \Later$ on delay modalities, and allowing $\tick[m']$ to eliminate 
modality $m$ if $m \leq m'$. This induces an inclusion 
\begin{equation}
  \label{eq:embed}
  \Varid{embed} = \lambda x. \delay(\adv\, x) : \Delay A \to \Later A
\end{equation}
for all $A$. In
general there is no inclusion in the opposite direction, except for a class of special types 
which we refer to as limit types, defined in Figure~\ref{fig:stable:and:limit}. The terminology
refers to the step indexed interpretation of types, see \autoref{sec:metatheory}.

The tokens $\tick[\Later]$ and $\tick[\Delay]$ are collectively
referred to as \emph{ticks} and the rules in Figure~\ref{fig:wf:cxt}
stipulate that there may be at most one tick in a context. This means
that programs can refer to data from the present and previous time
step, but not from earlier time steps. This is a crucial restriction
that rules out implicit space leaks, and similar restrictions can be
found in many other modal languages for
FRP~\citep{krishnaswami13frp,bahr2019simply,cave14fair}.

The second kind of token in Lively RaTT is $\lock$, which separates the context into static
variables to the left of $\lock$ and dynamic variables to the right. Static variables are 
time-independent whereas the dynamic ones can depend on reactive data available only
in the current instant. 
%
%At initialisation time the 
%variables of a program are instantiated, and the program is prepared to be run as a reactive
%program. At dynamic time it is executed against a stream of input to produce a stream of output.
This distinction is only made once, so there can be at most one $\lock$ in a context. 
The notion of time step is relevant only for dynamic variables, and therefore tokens $\tick[\Delay]$ and $\tick[\Later]$ 
can only appear to the right of a $\lock$. The rules for well-formed contexts can be found in 
Figure~\ref{fig:wf:cxt}.

The token $\lock$ is associated with the modality $\Box$. 
Data of type $\Box A$ should be thought of as stable data, i.e., data that does not depend on 
time-dependent dynamic data, and can thus be safely transported into the future without causing 
space leaks. This is reflected in the introduction rule for $\Box$  
which ensures that $\rbox\, t$ can not contain free dynamic variables 
(i.e. variables to the right of $\lock$), and in the elimination rule allowing $\hastype{\Gamma}{t}{\Box A}$ 
to be eliminated in a context $\Gamma, \lock, \Gamma'$ also when $\Gamma'$ contains a tick.
%$\tick[\Delay]$ or $\tick[\Later]$.

Stable types (Figure~\ref{fig:stable:and:limit}) are types whose values by nature cannot contain 
time-dependent data, and so can be used in any dynamic context. This is implemented in the 
language by allowing variables of stable types to be introduced also over tokens. Generalising this
to all variables would lead to space leaks. 
%\rasmus{Do we need to say that $\progress$ and $\promote$ can be encoded?}. 
Note in particular that function types
are not stable since closures can contain time-dependent data.

Elements of type $\Box A$ can also be constructed as guarded recursive fixed points. These are 
particularly useful for programming with guarded recursive types, i.e., types of the form 
$\Fix \;\alpha.A$ satisfying the type isomorphism $\Fix \;\alpha.A\iso A[\Later(\Fix \; \alpha.A)/\alpha]$.
Note that there is no restriction on $A$, which can in principle contain also negative occurrences of $\alpha$,
although we shall not be using that in this paper. The basic FRP types of streams and events 
can be encoded as guarded recursive types
\begin{align*}
 \Str A & \defeq \Fix \;\alpha.A\times \alpha &
 \Event A & \defeq \Fix \;\alpha.A + \alpha 
\end{align*}

The fixed point combinator as defined by \citet{nakano2000} is simply a term of type $(\Later A \to A) \to A$. 
In FRP a few adjustments
must be made to that. First of all, a fixed point will be called repeatedly at different dynamic times.
To avoid space leaks, fixed points should therefore not have free dynamic variables (although the recursion variable
itself should be dynamic), and the type of the fixed point should be of the form $\Box A$. In Simply RaTT, the typing rule
for fixed points states that $\hastype{\Gamma}{\fix \; x.t}{\Box A}$ if 
$\hastype{\Gamma,\lock, x : \Later A}{t}{A}$. In Lively RaTT this is too restrictive, since it prohibits
nesting of guarded fixed points and recursion over elements of the until-types $A \Until B$. 
The premise of the rule is therefore $\hastype{\Gamma,x : \Box\Later A,\lock}{t}{A}$, which gives a more general
fixed point rule. 

As an example, mapping of functions over streams can be defined using fixed points as 
\mapIntro
where $::$ refers to the infix constructor for streams, which in the example is also used for pattern matching. 
Note that the input function $f$ has type $\Box (A \to B)$ since it must be called at all futures. 

Lively RaTT features two kinds of inductive types. The first is the natural numbers with essentially the 
standard typing rules for $\zero$, $\suc$ and recursion. Note that these apply in any context $\Gamma$ independent
of which tokens are in $\Gamma$. The second is the until-type of LTL, to be thought of as the inductive 
solution to $A \Until B \iso B + A \times \Delay(A \Until B)$. As for the natural numbers, there is no restriction on the context for
the introduction rules, but the elimination rule is by nature dynamic, since elimination of an element of 
$A \times \Delay(A \Until B)$ should recurse one time step from now on the advanced element of type $A \Until B$. 
To avoid space leaks, the recursors should be stable, i.e., not depend on dynamic data. Thus eliminating from
$A \Until B$ into a type $C$ requires recursors of type $\Box (B \to
C)$ and $\Box(A \to \Delay (A \Until B) \to \Delay C \to C)$. 

Finally, note that Lively RaTT is a higher-order functional programming language with the restriction that lambda
abstraction is only allowed in contexts with no $\tick[\Delay]$ and $\tick[\Later]$. This restriction is inherited from
Simply RaTT where it is necessary to guarantee the lack of space leaks. As we shall see, this appears not
to be a limitation in practice. 

\section{Programming in Lively RaTT}
%\section{Examples}
\label{sec:examples}
%\patrick{Rename this section to ``Programming in Lively RaTT''?}

This section gives a number of examples of programming in Lively RaTT.
First, we give a series of examples of programming with events.
Secondly, we show how to encode \emph{fairness} and how to implement a
fair scheduler.

\subsection{Events and diamond}
As described in the introduction, events that \emph{may} occur can be encoded
in Lively RaTT as $\event*\,{A} \defeq \Fix\,\alpha.A + \alpha$. For
example, the event that loops forever can be defined as
\loopEvent

The type of events that \emph{must} occur can be encoded as the 
diamond modality from LTL, namely $\Diamond(A) \defeq 1 \Until A$.
Below we will use the following shorthand when working with $\event*$ and
$\Diamond$:\\
\begin{minipage}[c]{0.35\linewidth}
\shorthandDia
\end{minipage}
\begin{minipage}[c]{0.65\linewidth}
\shorthandEv
\end{minipage}
Here the $\now_\Diamond$ and $\now_{\event*}$ maps are like the
$\Varid{return}$ map from a monad. Both $\event*$ and $\Diamond$
further admit a map reminiscent of the $\Varid{bind}$ map of a monad.
For $\event*$ this is given by:
\evBind
where $\circledast$ is the infix notation of the delayed function call
as defined in \autoref{eq:delayApp}, which can be given the more
general type $m_1(A \to B) \to m_2(A) \to m_3(B)$ for
$m_i \in \{\Delay,\Later\}$ with $m_1,m_2 \leq m_3$. To see that
$\Varid{bind}_\sym{Ev}$ is well-typed, consider the two cases. In the
first case $a : A$ and hence, the unboxed $f$ can be applied
immediately. In the second case $e : \Later \event*\, A $ and
$b : \Box\Later (\event*\, A \to \event*\,B )$. It then follows that
$\unbox \, b : \Later (\event*\, A \to \event*\,B )$ and thus, by a
delayed function application,
$(\unbox \, b) \circledast e : \Later \event*\,B$. This is then
wrapped in $\wait_{\event*}$ to produce an element of $\event*\,B$ as
needed. Note the requirement for the map $f : A \to \event*\,{B}$ to
be stable. This is because it might be applied \emph{in the future}.

For $\Diamond$, we define the map
\diaBind
where again $f$ must be stable. To see that $bind_\Diamond$ is well
typed, consider the base and recursion case. In the base $a : A$ and
hence, the unboxed $f$ can be applied immediately. In the recursion
case $u : \Unit, w : \Delay(A \Until B)$ and $d : \Delay\Diamond{B}$,
hence also $\wait_\Diamond\,d : \Diamond{B}$ as required.

We will also use sugared syntax for recursive definitions, writing
e.g. the above definition of
$\Varid{bind}_\sym{Ev}$ as
\evBindSugared
The $\lock$ separates the variables into those received before and after
$\fix$, and since the two cases define $\Varid{bind}_\sym{Ev}\,f$ by guarded
recursion, this should be considered an atomic subexpression with type
$\Box\Later(\event*\,A \to \event*\,B )$.

Similarly, the definition of $\Varid{bind}_\Diamond$ can be written in
the sugared syntax as
\diaBindSugared
Here, the two cases of the recursive definition of $\Varid{bind}'_\Diamond$ 
are written as pattern matching syntax. In the second case the subterm
$\Varid{bind}^\prime_\Diamond \, d$ represents the recursive call and
should therefore be read as having type $\Delay\Diamond B$. To
elaborate such definitions back into $\sym{rec}_{\Until}$, replace
calls such as $\Varid{bind}^\prime_\Diamond \, d$ with a fresh
variable that represents the call to the recursor. We chose to use the
above style to make it clear when delayed arguments are used and how
they are passed around.

Since $\Diamond$ represents events that must occur, and $\event*$
represents more general, possibly occurring, events there is an
inclusion from $\Diamond$ to $\event*$. Using the above syntax, this
can be defined by $\Until$-recursion as
\diamondInclusion
This map makes crucial use of the fact that $\Delay$ is a sub-modality
of $\Later$ in the call to $\Varid{embed}$, as
defined in \autoref{eq:embed}. 

A further consequence of the sub-modality relation is that
non-terminating events ``overrule'' terminating events. Consider
$\event*$ containing a $\Diamond$:
\diamondEvent
The converse, a function with type
$\Diamond\event*\,A \to \Diamond\,A$, can not be written in the
language, since the inner event may be non-terminating.

There is in general no inclusion the from $\event*$ into $\Diamond$
because of the requirement that elements of $\Diamond A$ terminate.
One solution is to wrap the conversion in a timeout, which will handle
the non-terminating case. We must then supply a natural number,
representing how many time steps to wait, and let the conversion fail
if we go beyond that. We define by natural number recursion
\noindent
\begin{minipage}{1.0\linewidth}
  \timeoutEvent
\end{minipage}
Typing of the fourth case uses the requirement $A \, \limit$: The
$\delay$ corresponds to a $\Delay$-tick in the context, which in
general can not be used to advance $e : \Later \event*\,{A}$. It can in
this case since $A \, \limit$ implies $\Later\event*\,{A}\,\limit$.

Even though we can not give a general function $\Diamond\event* A \to
\Diamond A$, we can use the above $\Varid{timeout}$ to give a function
$\Diamond\event*A \to \Diamond (1 + A)$.
\eventDiamond

In general we cannot join two events of type $\Diamond A$ and
$\Diamond B$ to an event of type $\Diamond(A \times B)$, i.e., waiting
for both events to occur and pairing up the result. Doing so for
arbitrary types $A$ and $B$ may lead to space leaks: If the two events
occur at different times, the result of whichever event occurs first
would need to be buffered until the second one occurs. However, if $A$
and $B$ are stable, we can explicitly buffer the early event
occurrence, which allows us to implement the join. The implementation
relies on two auxiliary functions that differ only in which order
their arguments are given. We give only one of them, implemented by
$\Until$-recursion:%
\diamondJoinAux%
where $\odot$ is the infix function%
\progressApp where %
$m_1,m_2 \in \{\Delay,\Later\}$ and $m_1 \leq m_2$.  With the above,
we can now define the join by $\Until$-recursion:%
\diamondJoin

We now give a function constructing elements of $\Diamond$ by
buffering data for a given number of time steps.
%
%The $\Diamond$ type can be used to define function \emph{buffers} some
%term for a given time. Given some element of a stable term and a time to
%wait, the function should return the element at that time.
%
For this we need a type
%Before we can define the function we need an additional construction,
%namely a type 
of \emph{temporal natural numbers} that will serve as a
means to count time:
\begin{align*}
  \Nat_\Delay = \Diamond \Unit
\end{align*}
This type can be thought of as natural numbers, where the successor
operation requires one time step to compute. The zero and successor
can be encoded as:
\begin{align*}
  0_\Delay &= \now_\Diamond \unit \\
  \suc_\Delay\, n &= \wait_\Diamond\,n
\end{align*}

Any temporal natural number can be imported into the future by means
of $\Until$-recursion.

\importTNat

Given a natural number, we can convert it into a temporal natural
number by recursion on natural numbers:
\mkTimer
Intuitively speaking, given a natural number $n$, $\Varid{timer}\,n$
is a timer with $n$ ticks.

The buffer function takes a temporal natural number and requires
$A$ to be stable, for the input to be buffered.
%
%Using the temporal natural numbers we create the buffer that given
%an element of a stable type and a temporal natural numbers, returns
%the element when the timer hits zero:
%
\buffer

As a final example of working with $\Diamond$ we define a simple server.
First off we define the type of servers as 
\begin{align*}
  \sym{Server} := \Fix\,\alpha.\alpha \times (\sym{Req} \to
  (\Diamond\sym{Resp} \times \alpha))
\end{align*}
where $\sym{Req}$ and $\sym{Resp}$ are the types of requests and
responses, respectively. In each step, a server can receive at most one 
request, which must eventually give a response. In either case, the
server will return a new server in the next time step, with a possibly
updated internal state.

With the above, we can define a simple
server that given a string $s$ and a number $n$, returns $\pair{s}{m}$
after $n$ time steps, where $m$ is the number of requests received. We
set $\sym{Req} := \Nat \times \sym{String}$ and
$\sym{Resp} := \sym{String} \times \Nat$, and consider
$\sym{String}$ to be $\stable$. The server is defined by guarded recursion:

\rServerTwo

The server can be run and initialized with $0$:

\rServerRun

\subsection{Fair streams}
\label{sec:fair:streams}

A stream of type $\Str{A + B}$ will in each step produce either a
value of type $A$ or of $B$. For example, we can implement a scheduler
that interleaves two streams in an alternating fashion, dropping every
other element of either stream:
\altStr
$\Varid{altStr}$ and $\Varid{altStr'}$ are defined by mutual guarded
recursion. The former starts with taking an element from the second
stream whereas the latter starts with taking from the first
stream. This mutual recursive syntax translates to a single guarded
fixed point in the calculus via a standard tupling construction that
constructs both functions simultaneously:%
\altStrMut%
In the above definition, the variable $r$ is of type
\[
  \Box \Later ((\Str A \to \Str B \to \Str {A + B}) \times (\Str A \to \Str B \to \Str{A + B}))
\]
We then use the projections $\pi_i$ lifted to $\Later$ to access the
desired component from $\unbox\,r$:
\begin{align*}
  &\pi^\Later_i : \Later(A_1 \times A_2) \to \Later A_i\\
  &\pi^\Later_i= \lambda x . \delay (\pi_i(\adv\,x))
\end{align*}
Given the above tupling construction, $\Varid{altStr}$ is then defined
as $\rbox(\pi_1\,(\unbox\,\mathit{altStrMut}))$. From now on we will
use the mutual guarded recursion syntax with the understanding that it
can be turned into a single guarded fixed point by tupling.

The following function inhabits the same type as $\Varid{altStr}$, but
it only draws elements from the first stream, dropping the second
stream altogether:

\dropSnd

Following the work by \citet{cave14fair}, we can refine the type
$\Str{A + B}$ to a type $\Fair A B$, whose inhabitants will produce in
each step a value of type $A$ or of type $B$, in a fair
manner:
\[
  \Fair A B = \Fix\, \alpha. A \Until (B \times \Later (B \Until (A \times
  \alpha)))
\]
A term of type $\Fair A B$ may first produce some elements of type
$A$, but must after finitely many steps produce an element of type
$B$. It may continue to produce more elements of type $B$, but must
eventually produce an element of type $A$ and then continue in this
manner indefinitely. This required behaviour prevents us from
implementing $\Varid{dropSnd}$ to produce a fair stream of type
$\Fair A B$. On the other hand, we can re-implement $\Varid{altStr}$
to produce a fair stream as follows:

\altFair

To simplify programming with fair streams we define shortcut
constructors for the type $\Fair A B$. To this end we define the
following variant of the type $\Fair A B$
\[
  \Fair* B A = B \Until (A \times \Later \Fair A B)
\]
We now have that $\Fair A B$ unfolds to
$A \Until (B \times \Later \Fair* B A)$ and thus the two types
$\Fair A B$ and $\Fair* A B$ are isomorphic. Fair streams are
constructed by either staying with the first type $A$ or switching to
the second type $B$.

\noindent
\begin{minipage}{0.45\linewidth}
\staySwitch  
\end{minipage}
\hfill
\begin{minipage}{0.45\linewidth}
\staySwitchB
\end{minipage}

\noindent
From the types above one can immediately see that we can only $\stay$
with the same type finitely often -- indicated by the $\Delay$
modality -- whereas we can $\switch$ arbitrarily -- indicated by the
$\Later$ modality. Using the above shorthands, the $\Varid{altFair}$
function can thus be implemented more concisely as follows:

\altFairSug

The fair stream type $\Fair A B$ can be considered a special case of
the stream type $\Str{A+B}$ with additional liveness constraints. We
can always forget these constraints by converting a fair stream into a
normal stream:

\runFair

The function $\Varid{runFair}$ is defined by guarded recursion with
two nested $\Until$-recursions on the two nested $\Until$-types that
make up the fair stream type. Note that the two recursive calls
$\Varid{run}_1\,d$ and $\Varid{run}_2\,d$ produce a delayed stream of
type $\Delay(\Str{A+B})$. Therefore, we have to use $\Varid{embed}$ to
convert them to type $\Later(\Str{A+B})$.

We conclude with an example that implements a more interesting
interleaving of two streams into a fair stream, namely the fair
scheduler from \citet{cave14fair} that selects a progressively
increasing number of elements from the first stream for each time it
selects an element from the second stream:

\schedSugar

In particular $\unbox\, \Varid{sch}\, 0\, \Varid{as}\, \Varid{bs}$
produces a fair stream of the following form:
\[    
  B\quad A\quad A\quad B\quad A\quad A\quad A\quad B\quad A\quad
  A\quad A\quad A\quad B\quad A\quad A\quad A\quad A\quad A\quad B
  \dots
\]

The fair scheduler is implemented by guarded mutual recursion with a
nested $\Until$-recursion. The natural number is first turned into a
timer, which is then recursed over using the $\Varid{until}$
function. In each recursive step of $\Varid{until}$ -- corresponding
to a tick of the timer -- we select from the first stream. But once
the timer reaches $0_\Delay$, we switch to selecting from the second
stream, then immediately switch to selecting from the first stream
again, increment the counter $m$, and proceed by guarded recursion.

Note that we require $A$ and $B$ to be limit types so that in turn
$\Str A$ and $\Str B$ are limit types. The latter is needed in the
first clause of the $\Varid{until}$ function so that we may apply the
recursive call $\Varid{until}\,n$ of type
$\Delay (\Nat \to \Str A \to \Str B \to \Fair A B)$ to both $as$ and
$bs$, which are of type $\Later \Str A$ and $\Later \Str B$,
respectively.

\section{Operational semantics}
\label{sec:op:sem}

The operational semantics of Lively RaTT is divided into two parts: an
\emph{evaluation semantics} that captures the computational behaviour
at each time instant (\autoref{sec:big-step-operational}), and a
\emph{step semantics} that describes the dynamic behaviour of a Lively
RaTT program over time. We introduce the latter in two stages. At
first we only look at programs without external input
(\autoref{sec:small-step-oper}). Afterwards we extend the semantics to
account for programs that react to external inputs
(\autoref{sec:reactive-small-step}), e.g., terms of type
$\Box(\Str A \to \Str B)$, which continuously read inputs of type $A$
and produce outputs of type $B$. Along the way we give a precise
account of our main technical results, namely productivity,
termination, liveness, and causality properties of the operational
semantics, as well as the absence of implicit space leaks. To prove
the latter, the evaluation semantics is formulated using a store on
which external inputs and delayed computations are placed. At each
reduction step, the step semantics garbage collects all elements
of the store that are more than one step old, thereby avoiding
implicit space leaks by contruction.

\subsection{Evaluation semantics}
\label{sec:big-step-operational}

\begin{figure}
    \textbf{Call-by-value $\lambda$-calculus:}
    \begin{mathpar}
 \inferrule*
  {~}
  {\heval{v}{\sigma}{v}{\sigma}}
  \and
  \inferrule*
  {\heval {t} {\sigma} {v} {\sigma'}\\
    \heval {t'} {\sigma'} {v'} {\sigma''}}
  {\heval {\pair{t}{t'}} {\sigma} {\pair{v}{v'}} {\sigma''}}
  \and
  \inferrule*
  {\heval {t} {\sigma} {\pair{v_1}{v_2}} {\sigma'} \\ i \in \{1,2\}}
  {\heval {\pi_i(t)} {\sigma} {v_i} {\sigma'}}
  \and
  \inferrule*
  {\heval t {\sigma} v {\sigma'} \\ i \in \{1,2\}}
  {\heval {\interm_i(t)} {\sigma} {\interm_i(v)} {\sigma'}}
  \and
  \inferrule*
  {\heval {t} {\sigma} {\interm_i(v)} {\sigma'}\\
    \heval {t_i[v/x]} {\sigma'} {v_i} {\sigma''} \\ i \in \{1,2\}}
  {\heval {\caseterm{t}{x}{t_1}{x}{t_2}} {\sigma} {v_i} {\sigma''}}
  \and
  \inferrule*
  {\heval{t}{\sigma}{\lambda x.s}{\sigma'} \\
    \heval{t'}{\sigma'}{v}{\sigma''}\\
    \heval {s[v/x]}{\sigma''}{v'}{\sigma'''}}
  {\heval{t\,t'}{\sigma}{v'}{\sigma'''}}
  \end{mathpar}
  \textbf{Modalities, $\Until$-types, guarded recursion:}
  \begin{mathpar}
  \inferrule*
  {l = \allocate{\sigma} \\ \sigma \neq \nullstore}
  {\heval{\delay\,t}{\sigma}{l}{(\sigma,l \mapsto t)}}
  \and
  \inferrule*
  {\heval {t}{\eta_N}{l}{\eta_N'} \\
    \heval {\eta_N'(l)}{(\eta_N'\tick\eta_L)} {v}{\sigma'}}
  {\heval {\adv\,t}{(\eta_N\tick\eta_L)}{v}{\sigma'}}
  \and
  \inferrule*
  {\heval{t}{\nullstore}{\rbox\,t'}{\nullstore} \\
    \heval{t'}{\sigma}{v}{\sigma'} \\ \sigma \neq \nullstore}
  {\heval{\unbox\,t}{\sigma}{v}{\sigma'}} \and
%\end{mathpar}
%  \textbf{3: Inductive fragment}
%  \begin{mathpar}
    \inferrule*
    {\heval{t}{\sigma}{v}{\sigma'}}
    {\heval{\suc \, t}{\sigma}{\suc \, v}{\sigma'}}
    \and
    \inferrule*
    {\heval{n}{\sigma}{\zero}{\sigma'}
      \\ \heval{s}{\sigma'}{v}{\sigma''}}
    {\heval{\recN{s}{x}{y}{t}{n}}{\sigma}{v}{\sigma''}}
    \and
    \inferrule*
    {\heval{n}{\sigma}{\suc \, v}{\sigma'}
      \\ \heval{\recN{s}{x}{y}{t}{v}}{\sigma'}{v'}{\sigma''}
      \\ \heval{t[v/x,v'/y]}{\sigma''}{w}{\sigma'''}}
    {\heval{\recN{s}{x}{y}{t}{n}}{\sigma}{w}{\sigma'''}}
    \and
    \inferrule*
    {\heval{t}{\sigma}{v}{\sigma'}}
    {\heval{\now\,t}{\sigma}{\now\,v}{\sigma'}}
    \and
    \inferrule*
    {\heval{t_1}{\sigma}{v_1}{\sigma'}
      \\\heval{t_2}{\sigma'}{v_2}{\sigma''}}
    {\heval{\wait\,t_1\,t_2}{\sigma}{\wait\,v_1\,v_2}{\sigma''}}
    \and
    \inferrule*
    {\heval{u}{\sigma}{\now\,v}{\sigma'}
      \\\heval{s[v/x]}{\sigma'}{w}{\sigma''}}
    {\heval{\recU{x}{s}{x}{y}{z}{t}{u}}{\sigma}{w}{\sigma''}}
    \and
    \inferrule*
    {\heval{u}{\sigma}{\wait\,v_1\,v_2}{\sigma'}
      \\\heval{t[v_1/x,v_2/y,l/z]}{(\sigma',l\mapsto \recU{x}{s}{x}{y}{z}{t}{\adv(v_2)})}
      {v'}{\sigma''} \\ l = \allocate{\sigma'}}
    {\heval{\recU{x}{s}{x}{y}{z}{t}{u}}{\sigma}{v'}{\sigma''}} \and
%  \end{mathpar}
%\textbf{4: Guarded recursive fragment}
%\begin{mathpar}
  \inferrule*
  {\heval{t}{\sigma}{v}{\sigma'}}
  {\heval{\into\,t}{\sigma}{\into\,v}{\sigma'}}
  \and
  \inferrule*
  {\heval{t}{\sigma}{\into\,v}{\sigma'}}
  {\heval{\out\,t}{\sigma}{v}{\sigma'}}
  \and
  \inferrule*
  {\heval{t}{\nullstore}{\fix\,x.t'}{\nullstore} \\
    \heval{t'[\rbox(\delay(\unbox(\fix\;x.t')))/x]}{\sigma}{v}{\sigma'} \\
  \sigma \neq \nullstore}
  {\heval{\unbox\,t}{\sigma}{v}{\sigma'}}
\end{mathpar}
  \caption{Evaluation semantics.}
  \label{fig:machine}
\end{figure}

The \emph{evaluation semantics} is presented as a big-step operational
semantics in \autoref{fig:machine} and describes how a configuration
consisting of a term $t$ and a store $\sigma$ evaluates to a 
value $v$ and an updated store $\tau$ in the current time instant, denoted
$\heval t\sigma v\tau$. In the machine, unlike the surface language, terms
can contain locations $l$, to be thought of as locations in the store.
Formally, these range over a given set $\locs$ 
of locations divided into a countably infinite collection of namespaces
each consisting of countably infinitely many locations. 
The grammar below describes which terms of the calculus
are considered values:
\[
  v,w ::= \unit \mid \zero \mid \suc \, v \mid \lambda x.t \mid
  \pair{v}{w} \mid \interm_i\, v \mid \; \rbox\,t\mid \delay\,t\mid \fix \; x.t \mid
  l \mid \into\,v \; \mid \now \, v \mid \wait \, v \, w
\]

A store can be of one of three forms: 
$\sigma ::= \nullstore \mid \eta_L \mid \eta_N \tick \eta_L$. The 
null store $\nullstore$ is used for a special state of the machine
in which it can neither write to the store, nor read from it. In the other
two forms $\eta_L$ (the `later' heap) and $\eta_N$ (the `now' heap) 
are heaps, i.e., pairs of a namespace
and a finite mapping from the namespace to terms. In either of the two
latter cases, the evaluation semantics can update all heaps present
by allocating fresh locations and placing delayed computations in them, 
but it can only read from the $\eta_N$ heap. The notation $\sigma,l \mapsto t$
refers to an extension of the heap furthest to the right in $\sigma$, and 
$\allocate{-}$ is a function returning fresh location in the heap furthest 
to the right. 

%
%Here we also include \emph{input locations} $l$, which provide the
%interface for interaction with external inputs, but defer further
%discussion of input locations until \autoref{sec:reactive-small-step}.

The fragment of Lively RaTT consisting of the lambda calculus with sums,
products, and natural numbers is given a standard call-by-value
semantics. The non-standard parts of the semantics involve the three
modalities $\Box$, $\Delay$, and $\Later$; the recursion principle for
$\Until$ types; and the fixed point combinator.

The constructors for the three modalities -- $\rbox$ and $\delay$ --
have a call-by-name semantics and produce suspended
computations. Terms of the form $\rbox\,t$ are values consisting of
unevaluated terms $t$, which are only evaluated once they are consumed
by $\unbox$. Terms of the form $\delay\,t$ are computations that may
be executed in the next time step. These computations are suspended
and placed on the heap, since the evaluation semantics only describes
computations at the current time instant. In the next time instant,
these can computation will be executed when $\adv$ is applied to their
location. The safety of the step semantics shows that such delayed
computations will only be executed in the next time step, and do not
need to be stored for future steps. Note that in some special cases, a
term $\delay(t)$ will be executed immediately, for example when
evaluating a closed term of the form $\adv(\delay(t))$. Although such
a term is not well-typed as a closed term by itself, it can occur in
the evaluation of the step semantics.
%\rasmus{Is this confusing? Should we not mention this.}

The operational semantics of the guarded fixed point combinator $\fix$
closely follows the intuition provided by its type: The fixed point
$\fix\,x.t$ is unfolded by delaying it into the future and
substituting it for $x$ in $t$. However, since $\fix\,x.t$ is of type
$\Box A$, it first has to be unboxed before the delay and boxed again
afterwards.

The recursion principle for $\Until$ types is similar to the primitive
recursion principle one would obtain for an inductive type
$\mu \alpha. B + (A \times \alpha)$. The difference to an $\Until$
type, i.e., a type $\mu \alpha. B + (A \times \Delay \alpha)$, is that
each recursive call $\recU*(\dots)$ is shifted one time step into the
future by placing it in the heap. As opposed to $\fix$, however, no additional
unboxing and re-boxing is required.

\subsection{Step semantics}
\label{sec:small-step-oper}

\begin{figure}
  \begin{mathpar}
    \inferrule
    {\heval{t}{\eta\tick}{v :: w}{\eta_N\tick\eta_L}}
    {\state{t}{\eta} \forwards{v} \state{\adv\,w}{\eta_L}}
    \and
    \inferrule
    {\heval{t}{\eta\tick}{\wait\, v\, w}{\eta_N\tick\eta_L}}
    {\state{t}{\eta} \forwardu{v} \state{\adv\,w}{\eta_L}}
    \and
    \inferrule
    {\heval{t}{\eta\tick}{\now\, v}{\eta_N\tick\eta_L}}
    {\state{t}{\eta} \forwardu{v} \state{\mathsf{HALT}}{\eta_L}}    \and
    \inferrule
    {\state{t}{\eta} \forwardu{v} \state{t'}{\eta'}}
    { \state*{t}{\eta}{p} \forwardf{\interm_p\,v}
      \state*{t'}{\eta'}{p} }
    \and
    \inferrule
    {\state{t}{\eta} \forwardu{\pair{v}{w}} \state{\mathsf{HALT}}{\eta'}}
    {\state*{t}{\eta}{1} \forwardf{\interm_{2}\,v} \state*{\adv\,w}{\eta'}{2}}
    \and
    \inferrule
    {\state{t}{\eta} \forwardu{\pair{v}{w}} \state{\mathsf{HALT}}{\eta'}}
    {\state*{t}{\eta}{2} \forwardf{\interm_{1}\,v} \state*{\out(\adv\,w)}{\eta'}{1}}
  \end{mathpar}
  \caption{Step semantics for streams, until types and fair
    streams.}
  \label{fig:stream_machine}
\end{figure}

The \emph{step semantics} given in \autoref{fig:stream_machine}
describes the computation performed by a Lively RaTT program
\emph{over time}. The notation $\state t\eta \forwards{v} \state{t'}{\eta'}$ 
indicates the passage of one time step during which a configuration
consisting of a stream program $t$ and a heap $\eta$ 
transitions to the program $t'$ and heap $\eta'$ emitting the output $v$.  We give
three separate step semantics for stream, until, and fair stream
types, denoted $\forwards{}$, $\forwardu{}$, and $\forwardf{}$,
respectively. In addition, we state our metatheoretic results:
$\forwards{}$ is productive, $\forwardu{}$ is guaranteed to terminate,
and $\forwardf{}$ indeed produces a \emph{fair} stream. But we defer
the proofs of these results to \autoref{sec:metatheory}.

A closed term $t$ of type $\Box \Str A$ is meant to produce an
infinite stream $v_1, v_2, v_3, \dots$ of values of type $A$ in a
step-by-step fashion:
\[
  \state{\unbox\,t}{\emptyset} \forwards{v_1} \state{t_1}{\eta_1}
  \forwards{v_2} \state{t_2}{\eta_2} \forwards{v_3}
  \cdots
\]
Each step
$\state{t_i}{\eta_i} \forwards{v_i} \state{t_{i+1}}{\eta_{i+1}}$
proceeds by first evaluating $\state{t_i}{\eta_i\tick}$ to 
$\state{v_i :: w}{\eta'_i\tick \eta_{i+1}}$, i.e., the head
$v_i \colon A$ and the tail $w \colon \Later \Str A$ of the
steam. Computation may then proceed in the next time step with the
term $t_{i+1} = \adv\, w$ and heap $\eta_{i+1}$ The old heap
$\eta'_{i}$, which consists of $\eta_i$ possibly extended during
evaluation of $t_i$, is garbage collected.
%This
%application of the $\adv$ combinator signifies the shift to the next
%time step.

We can show that a term of type $\Box \Str A$ indeed produces such an
infinite sequences of outputs. To state the productivity property of
streams concisely, we restrict ourselves to streams over \emph{value
  types}, which are described by the following grammar:
\[
  U,V ::= \Unit \mid \Nat \mid U \times V \mid U + V
\]
\begin{theorem}[productivity]
  \label{thr:productivity}
  If $\hastype{}{t}{\Box\Str A}$, then there is
  an infinite sequence of reduction steps
    \[
      \state{\unbox\,t}{\emptyset} \forwards{v_1} \state{t_1}{\eta_1}
      \forwards{v_2} \state{t_2}{\eta_2} \forwards{v_3}
      \cdots
    \]
    Moreover, if $A$ is a value type, then $\hastype{}{v_i}{A}$ for all $i \ge 1$.
\end{theorem}
In particular, this means that the machine will never get stuck trying
to access a heap location that has been garbage collected. In other
words, the aggressive garbage collection of the heap used in the step
semantics is safe.

Intuitively speaking, value types describe static, time independent
data, and therefore exclude functions and modal types. Since $A$ is a
value type in the above theorem, we can give a concise
characterisation of the output produced by the stream in terms of the
syntactic typing $\hastype{}{v_i}{A}$.

The operational semantics for until types proceeds similarly to stream
types, but has an additional case for when the computation eventually
halts by evaluating to a value of the form $\now\,v$.

We can show that a term of type $\Box(A \Until B)$ produces a sequence
of values of type $A$, but eventually halts with a value of type $B$:
\begin{theorem}[termination]
  \label{thr:termination}
  If $\hastype{}{t}{\Box(A \Until B)}$,
  then there is a finite sequence of reduction steps
  \[
    \state{\unbox\,t}{\emptyset} \forwardu{v_1} \state{t_1}{\eta_1} \forwardu{v_2} \state{t_2}{\eta_2} \forwardu{v_2}
    \dots \forwardu{v_n} \state{\mathsf{HALT}}{\eta_n}
  \]
  Moreover, if  $A$ and $B$ are value types, then
  $\hastype{}{v_i}{A}$ for all $0< i<n$, and
  $\hastype{}{v_n}{B}$.
\end{theorem}

The computation performed by a fair stream of type $\Box\Fair A B$
requires a bit of additional bookkeeping: The state of the machine is
represented by a triple  $\state* t\eta p$ consisting of a term $t$, a heap
$\eta$ and a
value $p \in \set{1,2}$ that indicates the mode that the computation
represented by $t$ is in. If $p=1$, then the most recent output was of
type $A$, whereas $p = 2$ indicates that the most recent output was of
type $B$. A fair execution thus means that the computation may not
remain in the same mode indefinitely.
\begin{theorem}[liveness]
  \label{thr:fairness}
  If $\hastype{}{t}{\Box{\Fair A B}}$, then there is an infinite
  sequence of reduction steps
  \[
    \state*{\out\,(\unbox\,t)}{\emptyset}{1} \forwardf{\interm_{p_1} v_1} \state*{t_1}{\eta_1}{p_1}
    \forwardf{\interm_{p_2}v_2} \state*{t_2}{\eta_2}{p_2} \forwardf{\interm_{p_3}v_3}
    \dots
  \]
  such that for each $p\in \set{1,2}$, we have that $p_i = p$ for
  infinitely many $i \ge 1$.  Moreover, if $A$ and $B$ are value
  types, then $\hastype{}{\interm_{p_{i}}v_i}{A+B}$ for all $i \ge 1$.
\end{theorem}

As a special case of the fair stream type we obtain the type
$\sym{Live}(A) = \Fair \Unit A$. Terms of this type do not produce a
result at every time step but they will produce infinitely many
results.

\subsection{Reactive step semantics}
\label{sec:reactive-small-step}

\begin{figure}
  \begin{mathpar}
    \inferrule
    {\heval{t}{\eta,l\mapsto v :: l'\tick l'  \mapsto \unit}{v' ::
        w}{\eta_N\tick\eta_L,l' \mapsto \unit}\\l' = \allocate{\eta\tick}}
    {\stateS{t}{\eta}{l} \forwards*{v}{v'} \stateS{\adv\,w}{\eta_L}{l'}}
    \and
    \inferrule
    {\heval{t}{\eta,l\mapsto v:: l'\tick l' \mapsto \unit}{\wait\, v'\, w}{\eta_N\tick\eta_L,l\mapsto \unit}\\l' = \allocate{\eta\tick}}
    {\stateS{t}{\eta}{l} \forwardu*{v}{v'} \stateS{\adv\,w}{\eta_L}{l'}}
    \and
    \inferrule
    {\heval{t}{\eta,l\mapsto v::l'\tick l'\mapsto \unit}{\now\, v'}{\eta_N\tick\eta_L,l'\mapsto \unit}\\l' = \allocate{\eta\tick}}
    {\stateS{t}{\eta}{l} \forwardu*{v}{v'} \stateS{\mathsf{HALT}}{\eta_L}{l'}}    \and
    \inferrule
    {\stateS{t}{\eta}{l} \forwardu*{v}{v'} \stateS{t'}{\eta'}{l'}}
    { \stateF{t}{\eta}{l}{p} \forwardf*{v}{\interm_p\,v'}
      \stateF{t'}{\eta'}{l'}{p}}
    \and
    \inferrule
    {\stateS{t}{\eta}{l} \forwardu*{v}{\pair{v'}{w}} \stateS{\mathsf{HALT}}{\eta'}{l'}}
    {\stateF{t}{\eta}{l}{1} \forwardf*{v}{\interm_{2}\,v'} \stateF{\adv\,w}{\eta'}{l'}{2}}
    \and
    \inferrule
    {\stateS{t}{\eta}{l} \forwardu*{v}{\pair{v'}{w}} \stateS{\mathsf{HALT}}{\eta'}{l'}}
    {\stateF{t}{\eta}{l}{2} \forwardf*{v}{\interm_{1}\,v'} \stateF{\out(\adv\,w)}{\eta'}{l'}{1}}
  \end{mathpar}
  \caption{Reactive step semantics for streams, until types and fair
    streams.}
  \label{fig:stream_machine_react}
\end{figure}

The step semantics in \autoref{sec:small-step-oper} captures closed
computations without input from an external environment. To capture
reactive computations, we give for each step semantics
$\forward{M}{}$, a corresponding \emph{reactive} step semantics in
Figure~\ref{fig:stream_machine_react}. For instance, the reactive step
semantics for streams may at each step consume some input $v$ and
produce an output $v'$, which we denote by $\forwards*{v}{v'}$.

To supply input to the computation, the machine configurations for the
reactive step semantics are extended by an additional component $l$,
which is the heap location from where the next input value can be
retrieved. For instance, the reactive step semantics for streams takes
a step $\stateS{t}{\eta}{l}\forwards*{v}{v'}\stateS{t'}{\eta'}{l'}$ by
placing the value $v :: l'$ in the heap location location $l$, where
$l'$ is a freshly allocated heap location that serves as a stand-in
for the subsequent input. Assigning $l'$ the dummy value $\unit$ will
prevent the machine from allocating $l'$ during evaluation of $t$. The
reactive step semantics for $\Until$-types and fair streams follow the
same pattern.

Given a term $t$ of type $\Box(\Str A \to \Str B)$, and a sequence
$v_1,v_2,\dots$ of values of type $A$, there is an infinite sequence
of reduction steps
\[
  \stateS{\unbox\,t\,(\adv\,l_0)}{\emptyset}{l_0}  \forwards*{v_1}{v'_1} \stateS{t_1}{\eta_1}{l_1} 
  \forwards*{v_2}{v'_2} \stateS{t_2}{\eta_2}{l_2} \forwards*{v_3}{v'_3}
  \cdots
\]
such that $\hastype{}{v_i}{B}$ for all $i \ge 1$. The first term of
the computation sets up the initial promise of an input by giving the
term $\unbox\,t$ of type $\Str A \to \Str B$ the argument
$\adv\,l_0$. The location $l_0$ is simply the first fresh heap location, i.e.\
$l_0 = \allocate{\emptyset}$; While $\adv\,l_0$ is not a well-typed
term, it has the type $\Str A$ \emph{semantically}, in the sense that
$\adv\,l_0$ is a term in the logical relation $\dinterp{\Str A}$ that
we construct in \autoref{sec:metatheory}.

%The type system of Lively RaTT ensures that future inputs are not
%prematurely evaluated, i.e., the operational semantics will never try
%to read from a heap location $l$ that has been assigned the dummy
%value $\unit$. We shall show in \autoref{sec:metatheory} that this
%will not happen, and as a consequence Lively RaTT respects the
%principle of causality:
\begin{theorem}[causality]
  \label{thr:causality}
  Let $v_1,v_2,\dots$ be an infinite sequence of values with
  $\hastype{}{v_i}{A}$ for all $i \ge 1$.
  \begin{enumerate}[(i)]
  \item If $\hastype{}{t}{\Box(\Str A \to \Str B)}$, then there is an infinite
    sequence of reduction steps
    \[
    \stateS{\unbox\,t\,(\adv\,l_0)}{\emptyset}{l_0}  \forwards*{v_1}{v'_1} \stateS{t_1}{\eta_1}{l_1} 
    \forwards*{v_2}{v'_2} \stateS{t_2}{\eta_2}{l_2} \forwards*{v_3}{v'_3}
    \cdots
  \]
    Moreover, if $B$ is a value type, then $\hastype{}{v'_i}{B}$ for
    all $i \ge 1$.
    \label{item:causalityStr}
  \item If $\hastype{}{t}{\Box(\Str A \to B \Until C)}$, then there is a finite
    sequence of reduction steps
    \[
      \stateS{\unbox\,t\,(\adv\,l_0)}{\emptyset}{l_0}
      \forwardu*{v_1}{v'_1} \stateS{t_1}{\eta_1}{l_1} 
      \forwardu*{v_2}{v'_2} \stateS{t_2}{\eta_2}{l_2} \forwardu*{v_3}{v'_3}
      \dots \forwardu*{v_n}{v'_n} \stateS{\mathsf{HALT}}{\eta_n}{l_{n}}
      \]
    Moreover, if $B$ and $C$ are value types, then
    $\hastype{}{v'_i}{B}$ for all $0< i<n$, and
    $\hastype{}{v'_n}{C}$.
    \label{item:causalityUntil}
  \item If $\hastype{}{t}{\Box(\Str A \to \Fair B C)}$, then there is an infinite
    sequence of reduction steps
    \[
      \stateF{\out\,(\unbox\,t\,(\adv\,l_0))}{\emptyset}{l_0}{1}
      \forwardf*{v_1}{\interm_{p_1} v'_1} \stateF{t_1}{\eta_1}{l_1}{p_1}
      \forwardf*{v_2}{\interm_{p_2}v'_2} \stateF{t_2}{\eta_2}{l_2}{p_2}
      \forwardf*{v_3}{\interm_{p_3}v'_3} \dots
    \]
    such that for each $p\in \set{1,2}$, we have that $p_i = p$ for
    infinitely many $i \ge 1$. Moreover, if $B$ and $C$ are value
    types, then $\hastype{}{\interm_{p_{i}}v'_i}{B+C}$ for all
    $i \ge 1$.
    \label{item:causalityFair}
  \end{enumerate}
\end{theorem}

Since the operational semantics is deterministic, in each step
$\stateS{t_i}{\eta_i}{l_i} \forwards*{v_{i+1}}{v'_{i+1}}
\stateS{t_{i+1}}{\eta_{i+1}}{l_{i+1}}$ the resulting output $v'_{i+1}$
and new state of the computation
$\stateS{t_{i+1}}{\eta_{i+1}}{l_{i+1}}$ are uniquely determined by the
previous state $\stateS{t_i}{\eta_i}{l_i}$ and the input
$v_{i+1}$. Thus, $v'_{i+1}$ and
$\stateS{t_{i+1}}{\eta_{i+1}}{l_{i+1}}$ are independent of future
inputs $v_j$ with $j > i+1$. The same is true for the corresponding
reactive step semantics of $\Until$-types and fair streams.

\section{Metatheory}
\label{sec:metatheory}

In this section we show the soundness of the type system, which
typically means that a well-typed term will never get stuck. However,
we show a stronger, \emph{semantic} type soundness property that will
allow us to prove the operational properties detailed in
\autoref{sec:op:sem}.  To this end, we devise a Kripke logical
relation. Essentially, such a logical relation is a family
$\llbracket A \rrbracket(w)$ of sets of closed terms that satisfy the
desired soundness property. This family of sets is indexed by $w$
drawn from a suitable sets of ``worlds'' and defined inductively on
the structure of the type $A$ and world $w$. The proof of soundness is
then reduced to a proof that $\hastype{}{t}{A}$ implies
$t \in \llbracket A \rrbracket(w)$.

The full proofs of our results can be found in the accompanying
technical appendix.

\subsection{Worlds}
To a first approximation, the worlds in our logical relation contain
two ordinals $\alpha \leq \omega$ and $\beta < \omega \cdot 2$ and a store
$\sigma$. The two ordinals are used to
define the logical relation for recursive types, the first for
(temporal) inductive types and the latter for step-indexed guarded
recursive types. For guarded recursive types, we achieve this by
defining $\llbracket \Later A \rrbracket(\sigma,\alpha,\beta)$ in terms of
$\llbracket A \rrbracket(\sigma,\alpha,\beta')$ for strictly smaller
$\beta'$.  Since unfolding $\Fix\,\alpha.A$ introduces a $\Later$, 
the step index decreases for the recursive call. For
the inductive types, we define
$\llbracket A \Until B \rrbracket(\sigma,\alpha,\beta)$ in terms of
$\llbracket \Delay (A \Until B) \rrbracket(\sigma,\alpha',\beta)$ where
$\alpha' < \alpha$. Intuitively, $\alpha$ gives an upper limit to the
number of unfoldings of the inductive type used in terms.
%The interpretation of $A \Until B$
%at $\omega$ is be the union of the interpretation at all finite
%($\alpha$)-indices.

While this setup is sufficient for proving productivity, safety of
garbage collection, termination, and liveness properties, it is not
enough to capture causality. To characterise causality, the logical
relation also needs to account for the possible inputs a given term
may receive. We do this by further indexing the logical relation by a
sequence of future inputs. To this end, we assume an infinite sequence
of heaps $\eta_1;\eta_2;\ldots$, denoted $\ol{\eta}$, that describes
the input that is received at each point in the future. The namespaces
of all heaps in $\ol\eta$ and $\sigma$ are assumed pairwise disjoint.
The worlds in our logical relation are thus of the from
$(\sigma,\ol{\eta},\alpha,\beta)$.

As is standard for Kripke logical relations, our relation will be
closed under moving to a bigger world. Worlds are ordered as follows:
$(\sigma,\ol{\eta},\alpha,\beta) \leq (\sigma',\ol{\eta}',\alpha',\beta')$ if
$\sigma \tickle \sigma'$, $\ol\eta\heaple \ol\eta'$, $\alpha =  \alpha'$ and
$\beta' \leq\beta$. Here $\ol\eta\heaple \ol\eta'$ refers to the pointwise 
ordering on partial maps (assuming identity of namespaces) and
$\sigma \tickle \sigma'$ is the extension of this with the rule 
$\eta_L \tickle \eta_N\tick \eta_L$ to a preorder. Note that the null store
$\nullstore$ is only related to itself under this ordering.

\subsection{Support and renamings}
To prove closure of the Kripke semantics under store extensions, 
a similar property must be proved for the machine, i.e, if $\state t\sigma$
evaluates to a value and if $\sigma\tickle \sigma'$ then $\state t{\sigma'}$
evaluates to the same value. However, this statement is not entirely
true, because the machine, when run in an extended state, may allocate 
different locations on the heap and the resulting values may differ correspondingly.
To prove that this is the only way that the values can differ, 
we introduce notions of a renaming and support. 

A renaming is a map 
$\phi: \locs \to \locs$ respecting name spaces. Such a map acts on terms
by substitution, and this extends to heaps and stores by 
$\phi(\eta, l \mapsto t) = \phi(\eta), \phi(l) \mapsto \phi(t)$. We write
$\phi : (t, \sigma, \ol\eta) \to (t', \sigma', \ol\eta)$ if $\phi(t) = t', 
\phi(\sigma) \tickle \sigma', \phi(\ol\eta) \heaple \ol\eta'$.
%
%In the operational semantics of Lively RaTT we assume a deterministic
%allocator function, that creates fresh labels. When
%we extend stores and receive external inputs, we need to ensure that no
%labels are allocated twice. First of, we require that all heaps, both
%in stores and future input, are equipped with their own
%namespace. This ensure that the same label can not be allocated in two
%different heaps. When working in a given namespace, we can solve name
%clashes by renaming the offending label. To formalize this idea, we
%introduce the notions of \emph{support} and \emph{renaming}. 
%
Given a term $t$ and a pair of a store and a heap sequence
$(\sigma,\ol{\eta})$, we say that $t$ is \emph{supported} by
$(\sigma,\ol{\eta})$, written $t \support (\sigma,\ol{\eta})$, if
whenever a location in $t$ occurs in the namespaces of $\sigma$ or
$\ol{\eta}$, it must be in the domain of $\sigma$ or $\ol{\eta}$,
respectively. Given $(\sigma,\ol{\eta})$, we write
$(\sigma,\ol{\eta}) \; \supported$ if all values in the codomains of
$\sigma$ and $\ol{\eta}$ are supported by $(\sigma,\ol{\eta})$. 
We restrict attention to Kripke worlds where $(\sigma, \ol\eta)$ is
supported. 
%A renaming is a map from labels to labels, which changes only a finite
%subset of labels. We write
%$\phi : (t,\sigma,\ol{\eta}) \to (t',\sigma',\ol{\eta}')$ to denote
%the renaming from $(t,\sigma,\ol{\eta})$ to
%$(t',\sigma',\ol{\eta}')$, where we allow $\sigma'$ and $\ol{\eta}'$
%to be extensions of $\sigma$ and $\ol{\eta}$, respectively.

\subsection{Logical relation}
Our logical relation consists of two parts: A \emph{value relation}
$\vinterpN{A}{w}$ that contains all values that semantically inhabit
type $A$ at the world $w$, and a corresponding \emph{term relation}
$\mathcal{T}\llbracket A \rrbracket(w)$ containing terms that evaluate
to element in $\vinterpN{A}{w}$. The two relations
are defined by mutual induction in \autoref{fig:logical-relation}.
More precisely, the two relations are defined by well-founded
recursion by the lexicographic ordering on the tuple
$(\beta,\tsize{A},\alpha,e)$, where $\tsize{A}$ is the size of $A$ defined
below, and $e = 1$ for the term relation and $e = 0$ for the value
relation.
\begin{align*}
  \tsize{\alpha} &= \tsize{\Later A} = \tsize{1} = \tsize{\Nat} =  1
  \\
  \tsize{A \times B} &= \tsize{A + B} = \tsize{A \Until B} = \tsize{A \to B} = 1 +
                       \tsize{A} + \tsize{B}
  \\
  \tsize{\Box A} &= \tsize{\Delay A} =  \tsize{\Fix\, \alpha. A} =  1 +
                   \tsize{A}
\end{align*}
Note that since $\tsize{\Later (\Fix\;\alpha.A)} = \tsize \alpha$, the
size of $A[\Later\Fix\,\alpha.A/\alpha]$ is strictly smaller than that
of $\Fix\,\alpha.A$, which justifies the well-foundedness of recursive
types. Note also that
$\vinterpN{A \Until B}{\sigma,\ol{\eta},\alpha,\beta}$ is defined in terms
of $\vinterpN{\Delay(A \Until B)}{\sigma,\ol{\eta},\alpha',\beta}$ for
$\alpha' < \alpha$. To obtain well-foundedness, we would need
$\tsize{\Delay(A\Until B)} \leq \tsize {A\Until B}$, which is not
true. But this problem can be avoided by ``inlining'' the definition
of $\vinterpN{\Delay(A \Until B)}{\sigma,\ol{\eta},\alpha',\beta}$, which
is defined in terms of
$\tinterpN{A\Until B}{\sigma',\ol{\eta}',\alpha',\beta}$ where $\sigma'$
and $\ol{\eta}'$ are appropriately modified. For the sake of readability,
we will keep the definitions as given.

%The definition of $\Unit, \Nat, A \times B$ and $A + B$ is standard.

\begin{figure}
  \begin{align*}
    \vinterpN{1}{w}
    =\ &\set{\unit},
    \\
    \vinterpN{\Nat}{w}
    =\ &\setcom{\suc^n \, \zero}{ n \in \nats},
    \\
    \vinterpN{A \times B}{w}
    =\ &\setcom{(v_1,v_2)}{v_1 \in \vinterpN{A}{w}
      \land v_2 \in \vinterpN{B}{w}},
    \\
    \vinterpN{A + B}{w}
    =\ &\setcom{\interm_1 \, v}{v \in
      \vinterpN{A}{w}} \cup
      \setcom{\interm_2 \, v}{v \in
      \vinterpN{B}{w}}
    \\
    \vinterpN{A \to B}{\sigma,\ol{\eta},\alpha,\beta}
    =\ &\{ \lambda x.t \mid
      t \support (\sigma,\ol{\eta}) \land
      \forall\beta'\le\beta.
      \forall \psi : (t,\sigma,\ol{\eta}) \to (t',\sigma',\ol{\eta}').\\
%     & \hspace{45pt} 
     % (\gc{\sigma'},\ol{\eta}')\, \supported \Rightarrow \\
    &\hspace{45pt}\forall v \in \vinterpN{A}{\gc{\sigma'},\ol{\eta}',\alpha,\beta'}.
      t'[v/x] \in \tinterpN{B}{\gc{\sigma'},\ol{\eta}',\alpha,\beta'}\}
    \\
    \vinterpN{\Box A}{\sigma,\ol{\eta},\alpha,\beta}
    =\ &\setcom{t}{
      \forall \ol{\eta}'.
      \unbox\,t \in \tinterpN{A}{\emptyset,\ol{\eta}',\alpha,\beta} \land t \, \labelfree}
    \\
    \vinterpN{\Delay A}{\sigma,(\eta;\ol{\eta}),\alpha,\beta}
    =\
      &\begin{cases}
        \dom{\gc{\sigma}} & \beta = 0 \land \sigma \neq \nullstore
        \\
        \setcom{l}{\adv \, l \in
          \tinterpN{A}{\gc{\sigma}\tick\eta,\ol{\eta},\alpha,\beta'}} &
        \beta = \beta' + 1 \land \sigma \neq \nullstore
        \\
        \setcom{l}{\adv \, l \in
          \tinterpN{A}{\gc{\sigma}\tick\eta,\ol{\eta},\alpha,\beta}} &
        \beta \text{ limit ordinal} \land \sigma \neq \nullstore
      \end{cases} 
    \\
    \vinterpN{\Later A}{\sigma,(\eta;\ol{\eta}),\alpha,\beta}
    =\ 
      &\begin{cases}
        \dom{\gc{\sigma}} & \beta = 0 \land \sigma \neq \nullstore
        \\
        \setcom{l}{\adv \, l \in
          \tinterpN{A}{\gc{\sigma}\tick\eta,\ol{\eta},\alpha,\beta'}} &
        \beta = \beta' + 1 \land \sigma \neq \nullstore
        \\
        \bigcap_{\beta' < \beta} \vinterpN{\Later
          A}{\sigma,(\eta;\ol{\eta}),\alpha,\beta'} & \beta \text{
          limit ordinal} \land \sigma \neq \nullstore
      \end{cases}
    \\
    \vinterpN{\Fix\,\alpha.A}{w}
    =\ &\setcom{ \into(v)}{v \in
      \vinterpN{A[\Later(\Fix\,\alpha.A)/\alpha]}{w}}
    \\
    \vinterpN{A \Until B}{\sigma,\ol{\eta},\alpha,\beta}
    =\ &\setcom{\now \,v}{v \in
      \vinterpN{B}{\sigma,\ol{\eta},\omega,\beta}} \cup
    \\
     &\setcom{\wait \,v\,w}{v \in
      \vinterpN{A}{\sigma,\ol{\eta},\omega,\beta} \land
      \exists \alpha' < \alpha. w \in \vinterpN{\Delay(A \Until B)}{\sigma,\ol{\eta},\alpha',\beta}}
    \\
    \tinterpN{A}{\sigma,\ol{\eta},\alpha,\beta}
    =\ &\setcom{t}{t \support (\sigma,\ol{\eta}) \land \exists \sigma', v
      . \heval{t}{\sigma}{v}{\sigma'}
      \land v \in \vinterpN{A}{\sigma',\ol{\eta},\alpha,\beta}}
  \end{align*}
  ~\\[-0.5em]
  \textsc{Garbage collection:}\qquad $\gc{\nullstore} = \nullstore\qquad\gc{\eta_L} =
  \eta_L\qquad \gc{\eta_N\tick\eta_L} = \eta_L$
  \caption{Logical Relation.}
  \label{fig:logical-relation}
\end{figure}
In the definition for $A \to B$ we explicitly add the closure
properties for support, renaming and worlds. Further, we restrict the
lambda abstractions to garbage collected stores. This reflects the
typing rule for lambda abstractions, which requires $\Gamma$ to be
tick-free. 

The definition of $\Box A$ captures the notion of stability. All terms
must be free of locations and able to evaluate safely with any input
sequence and hence, in any future.

The value relations for $\Later A$ and $\Delay A$ 
%encapsulates the passage
%of one time step and thus requires that the promised input for the
%next time step indeed arrives. Their definitions 
differ only in the case where $\beta$ is a limit ordinal. In the
successor case, they encapsulate the soundness of garbage
collection: The set
$\vinterpN{m\; a}{\sigma,(\eta;\ol{\eta}),\alpha,\beta+1}, m \in
\{\Delay,\Later\}$ contains all heap locations that can be read and
executed safely in the next timestep. In particular, such terms must 
evaluate in the garbage collected store extended with the
next set of external inputs. If $\beta$ is a limit ordinal, $\Delay A$
has the same interpretation except that the index $\beta$ is fixed. This
is to ensure that inductive types have the correct global behaviour. 
On the other hand, $\Later A$ is defined to be the intersection
of the interpretation at all smaller ($\beta$)-indices.
%, which forces $\Later A$ to be a $\limit$ type. 
This definition is needed for the
interpretation of fixed points.

In the definition of $A \Until B$ we see the use of the $\alpha$-index to give
an upper bound on the number of unfoldings used in the elements of the 
logical relation.
%Here we explicitly count down in the case for $\wait$ to ensure the
%termination of the values in $\vinterpNew[\alpha,\beta]{A \Until B}{I}$. 
In
particular, if $\alpha = 0$, the relation contains only values of the
form $\now\,v$ whereas if $\alpha > 0$, the relation also contain values
of the form $\wait\,u\,w$ defined in terms of values from
$\vinterpN{\Delay(A\Until B)}{\sigma,\ol{\eta},\alpha',\beta}$ where $\alpha' < \alpha$.

Our value and term interpretation is closed w.r.t the
Kripke structure on $\sigma,\ol{\eta}$ and $\beta$ and the value relation is upwards closed
w.r.t $\alpha$ for $\Until$-types. 
\begin{lemma}[Kripke Properties]
  \label{lem:kripke}
  Given $A,B$ and worlds $(\sigma,\ol{\eta},\alpha,\beta), (\sigma',\ol{\eta}',\alpha',\beta')$ s.t 
  $\sigma \tickle \sigma', \ol{\eta} \heaple \ol{\eta}', \alpha \leq\alpha'$ and $\beta' \leq \beta$ we have
  \begin{enumerate}
  \item $\vinterpN{A}{\sigma,\ol{\eta},\alpha,\beta} \subseteq \vinterpN{A}{\sigma',\ol{\eta}',\alpha,\beta'}$
  \item $\tinterpN{A}{\sigma,\ol{\eta},\alpha,\beta} \subseteq \tinterpN{A}{\sigma',\ol{\eta}',\alpha,\beta'}$
  \item $\vinterpN{A \Until B}{\sigma,\ol{\eta},\alpha,\beta} \subseteq \vinterpN{A \Until B}{\sigma,\ol{\eta},\alpha',\beta}$
  \end{enumerate}
\end{lemma}

As stated above we treat $\Delay$ as a sub-modality of $\Later$ and
this is expressed semantically in the following lemma:
\begin{lemma}[Sub-modality]
  \label{lem:submodal}
  Given $A$ and world $(\sigma,\ol{\eta},\alpha,\beta)$ then
  \begin{align*}
    \vinterpN{\Delay A}{\sigma,\ol{\eta},\alpha,\beta} \subseteq \vinterpN{\Later A}{\sigma,\ol{\eta},\alpha,\beta}
  \end{align*}
\end{lemma}
\begin{proof}
  Follows by transfinite induction on $\beta$.
\end{proof}

The next lemma justifies the terminology `limit types', by showing
that the interpretation of these at limit ordinals is the intersection
of the interpretations at the ordinals below. In category theoretic
terms, the intersection is a limit, and such a type is a
sheaf~\cite{maclane2012sheaves}.

\begin{lemma}[Limit Types]
  \label{lem:limit_types}
  If $A\,\limit$ and $\beta$ is a limit ordinal, then 
  \begin{align*}
   \bigcap_{\beta' < \beta} \vinterpN{A}{\sigma,\ol{\eta},\alpha,\beta'} & = \vinterpN{A}{\sigma,\ol{\eta},\alpha,\beta}  &
   \bigcap_{\beta' < \beta} \tinterpN{A}{\sigma,\ol{\eta},\alpha,\beta'} & = \tinterpN{A}{\sigma,\ol{\eta},\alpha,\beta}
  \end{align*}
\end{lemma}
\begin{proof}
  In the first equality, the inclusion from right to left follows from
  \autoref{lem:kripke}, and the other inclusion is proved by induction
  on $A$. The second equality then follows from the first.
\end{proof}

In the special case where $A$ is a limit type, $\Delay$ and $\Later$
do in fact coincide:
\begin{corollary}[Sub-modality at limit]
  Given $A$ and a world $w$ s.t. $A \, \limit$, then
  \begin{align*}
    \vinterpN{\Delay A}{w} =
    \vinterpN{\Later A}{w}
  \end{align*}
\end{corollary}
\begin{proof}
  One inclusion always holds by \autoref{lem:submodal}, and the two
  sets are equal by definition except when $\beta$ is a limit
  ordinal. In that case, by \autoref{lem:limit_types} it suffices to
  show that if
  $v \in \vinterpN{\Later A}{\sigma,(\eta;\ol{\eta}),\alpha,\beta}$ then
  $\adv\;v \in \tinterpN{A}{(\gc{\sigma}\tick\eta);\ol{\eta},\alpha,\beta'}$
  for all $\beta'<\beta$, which follows from
  $v \in \vinterpN{\Later A}{\sigma,(\eta;\ol{\eta}),\alpha,\beta'+1}$
%  Follows by the fact that $\Later A\,\limit$, \autoref{lem:submodal}
%  and \autoref{lem:limit_types}.
\end{proof}

Finally, we obtain the semantic soundness of the language phrased as the following
fundamental property of the logical relation $\tinterpN{A}{\sigma,\ol{\eta},\omega,\beta}$.
\begin{theorem}[Fundamental Property]
  \label{thr:lrl}
  If $\hastype{\Gamma}{t}{A}$ and $\gamma \in \cinterpN{\Gamma}{\sigma,\ol{\eta},\beta}$,
  then $t\gamma \in \tinterpN{A}{\sigma,\ol{\eta},\omega,\beta }$. 
\end{theorem}
Here $\cinterpN{\Gamma}{\sigma,\ol{\eta},\beta}$ refers to the logical
relation for typing contexts defined in
\autoref{fig:context_relation}. Note the cases for $\Gamma,\tick[m]$,
which captures the intuition that variables occurring before
$\tick[m]$ arrive one time step before those to the right. Again
$\Delay$ and $\Later$ differ only when $\beta$ is a limit
ordinal. Cases not mentioned in the figure (such as
$\cinterpN{\cdot}{\sigma,\ol{\eta},\beta}$ for $\sigma\neq\nullstore$)
are interpreted as the empty set.
The theorem is proved by a lengthy but entirely standard induction on
the typing relation $\hastype{\Gamma}{t}{A}$.

%where the variables in $\Gamma$ must
%arrive one time step before, and hence we count up on $\beta$.

\begin{figure}
 \begin{align*}
   \cinterpN{\cdot}{\nullstore,\ol{\eta},\beta}
   &= \set{\ast}
   \\
   \cinterpN{\Gamma,x:A}{\sigma,\ol{\eta},\beta}
   &= \setcom{\gamma[x \mapsto v]}
     {\gamma \in \cinterpN{\Gamma}{\sigma,\ol{\eta},\beta},v
     \in\vinterpN{A}{\sigma,\ol{\eta
     },\omega,\beta}}
   \\
   \cinterpN{\Gamma,\tick[\Later]}{(\eta_N\tick\eta_L),\ol{\eta},\beta}
   &= \cinterpN{\Gamma}{\eta_N,(\eta_L;\ol{\eta}),\beta+1}
   \\
   \cinterpN{\Gamma,\tick[\Delay]}{(\eta_N\tick\eta_L),\ol{\eta},\beta}
   &=
     \begin{cases}
     \cinterpN{\Gamma}{\eta_N,(\eta_L;\ol{\eta}),\beta} &  \beta \text{
       limit ordinal} \\
     \cinterpN{\Gamma}{\eta_N,(\eta_L;\ol{\eta}),\beta+1} & \text{otherwise}
     \end{cases}
   \\
   \cinterpN{\Gamma,\lock}{\sigma,\ol{\eta},\beta}
   &= \bigcup_{\ol{\eta}'}\cinterpN{\Gamma}{\nullstore,\ol{\eta}',\beta}
   \hspace{15mm} \sigma \neq \nullstore
 \end{align*}
 \caption{Context Relation}
 \label{fig:context_relation}
\end{figure}

%Crucially, we use the
%fact that the logical relation is closed under both the Kripke
%properties and garbage collection and that semantically, $\Delay$ is a
%sub-modality of $\Later$.
As an easy consequence of the fundamental property and the fact the
empty substitution is an element of
$\cinterpN{\cdot,\lock}{\sigma,\ol\eta,\beta}$ for any store $\sigma$
and input sequence $\ol\eta$, we have the following property that we
shall use to prove Lively RaTT's operational properties:
\begin{corollary}[Fundamental Property]
  \label{cor:lrl}
  If $\hastype{\lock}{t}{A}$, then
  $t \in \tinterpN{A}{\sigma,\ol\eta,\omega,\beta}$ for all
  $\sigma, \ol\eta$ and $\beta$.
\end{corollary}

\subsection{Productivity, termination, liveness \& causality}
\label{sec:prod-term-}

In this section we demonstrate how we apply the fundamental property
of the logical relation to prove the operational properties of Lively
RaTT that we presented in \autoref{sec:small-step-oper} and
\autoref{sec:reactive-small-step}. We have formulated these
operational properties in terms of value types, so that we can use the
following correspondence between semantic and syntactic typing:
\begin{lemma}
  \label{lem:vinterp_value}
  Given any world $w$, value type $A$, and value $v$, we have that
  $v\in \vinterpN{A}{w}$ iff  $\hastype{}{v}{A}$.
\end{lemma}
\begin{proof}
  By a straightforward induction on $A$.
\end{proof}

\subsubsection{Productivity}

We start with the productivity property of streams of type $\Str A$.

Given a type $A$ and ordinal $\beta$, we define the following set of
machine configurations for $\forwards{}$ that are safe according to the
logical relation:
\[
  S(A,\beta) = \setcom{\state{t}{\eta}}{t \in \tinterpN{\Str
      A}{\eta\tick,\ol{\emptyset},\omega,\beta}}
\]
where we use the notation $\ol\emptyset$ for the sequence of empty
heaps with the appropriate namespace.

Intuitively speaking, a machine configuration $c$ in $S(A,\beta)$ will
be safe to execute for the next $\beta$ steps of $\forwards{}$ and
produces output of type $A$. We formulate the essence of the
productivity property of such a stream as follows:
\begin{lemma}[productivity]
  \label{lem:productivity}
  If $\state{t}{\eta} \in S(A,\beta+1)$, then there are
  $\state{t'}{\eta'} \in S(A,\beta)$ and
  $v \in \vinterpN{A}{\eta',\ol\emptyset,\omega,\beta+1}$ such that
  $\state{t}{\eta} \forwards{v} \state{t'}{\eta'}$.
\end{lemma}
In each step of a stream computation, we count down by one on the step
index $\beta$ and produce an output $v$ of semantic type $A$:
\begin{proof}[Proof of \autoref{thr:productivity}]
  By \autoref{cor:lrl} we have that
  $\state{\unbox\,t}{\emptyset} \in S(A,\beta)$ for any $\beta$. Using
  \autoref{lem:productivity}, we can thus extend any finite reduction
  sequence
  \[
    \state{\unbox\,t}{\emptyset} \forwards{v_1} c_1
    \forwards{v_2} c_2 \forwards{v_3}
    \cdots \forwards{v_n} c_n
  \]
  by an additional reduction step $c_n \forwards{v_{n+1}}
  c_{n+1}$. Since $\forwards{}$ is deterministic, this uniquely
  defines the desired infinite reduction. Moreover, given that $A$ is
  a value type, $\hastype{}{v_i}{A}$ follows for all $i\ge 1$ by
  \autoref{lem:vinterp_value}.
\end{proof}

\subsubsection{Termination}
\label{sec:termination}

Analogously to the set of machine states $S(A,\beta)$ for stream
types, we define the following set $U(A,B,\alpha,\beta)$ for
until types:
\[
  U(A,B,\alpha,\beta) = \setcom{\state{t}{\eta}}{t \in \tinterpN{A\Until
      B}{\eta\tick,\ol{\emptyset},\alpha,\beta}}
\]

This definition allows us to state the essence of the termination
property for until types as follows:
\begin{lemma}[termination]
  \label{lem:termination}
  Given $\state{t}{\eta} \in U(A,B,\alpha,\beta)$, one of the
  following two statements holds:
  \begin{enumerate}[(a)]
  \item There are $t'$, $\eta'$, $\alpha' < \alpha$, and $v \in \vinterpN{A}{
      \eta',\ol\emptyset,\omega,\beta}$ such that
    \[
      \state{t}{\eta} \forwardu{v} \state{t'}{\eta'} \text{
        and if } \beta > 0 \text{ then } \state{t'}{\eta'} \in
      U(A,B,\alpha',\beta-1)
    \]
    where $\beta-1 = \beta'$ if $\beta = \beta'+1$ and otherwise
    $\beta-1=\beta$. 
    \label{item:terminationA}
  \item There is some
    $v \in \vinterpN{B}{\eta',\ol\emptyset,\omega,\beta}$ such that
    $\state{t}{\eta} \forwardu{v} \state{\mathsf{HALT}}{\eta'}$.
    \label{item:terminationB}
  \end{enumerate}
\end{lemma}

\autoref{thr:termination} is now an easy consequence of the above
lemma and the fundamental property of the logical relation.
\begin{proof}[Proof of \autoref{thr:termination}]
  By \autoref{cor:lrl} $\unbox\, t \in U(A,B,\omega,\omega)$, and by
  \autoref{lem:termination} we can construct the desired sequence of
  reductions. Since the index $\alpha$ strictly decreases each time we
  take a step of the form \ref{item:terminationA}, the sequence must
  eventually terminate with a step of the form
  \ref{item:terminationB}. Moreover, by \autoref{lem:vinterp_value},
  the output values $v_i$ have the desired type given that $A$ and $B$
  are value types.
\end{proof}

\subsubsection{Liveness}
\label{sec:fairness}

Recall that the step semantics of fair streams $\forwardf{}$ is a
machine whose configurations are tuples $\state* t \eta p$, where
$p\in\set{1,2}$ indicates the current mode of the computation. The
behaviour of the different modes is captured by the following
definition of the set $F(A,B,\alpha,\beta)$ of such pairs:
\begin{align*}
  F(A,B,\alpha,\beta) =\
  &\setcom{\state*{t}{\eta}{1}}{\state{t}{\eta} \in U(A, B\times \Later (B \Until (A \times \Later \Fair A
    B)),\alpha,\beta)}
  \\
  \cup &\setcom{\state*{t}{\eta}{2}}{\state{t}{\eta} \in U(B,A \times \Later \Fair A
    B,\alpha,\beta)}
\end{align*}
That is, if $p = 1$, then $t$ belongs semantically to an until type
$A \Until (B\times \Later \Fair* B A)$, and otherwise $t$ belongs to
$B \Until (A \times \Later \Fair A B)$.

With this characterisation, we can formulate the essence of the
liveness property for fair streams:
\begin{lemma}[liveness]
  \label{lem:fairness}
  Given $\state*{t}{\eta}{p} \in F(A,B,\alpha,\beta)$, one of
    the following statements is true:
     \label{item:livenessII}
    \begin{enumerate}[(a)]
    \item \label{item:livenessA} there are $t', \eta', \alpha' < \alpha$, and
      $v \in \vinterpN{A+B}{\eta',\ol\emptyset,\omega,\beta}$ such that
      \[
        \state*{t}{\eta}{p} \forwardf{v} \state*{t'}{\eta'}{p} \text{ and }
        \state*{t'}{\eta'}{p} \in F(A,B,\alpha',\beta') \text{ for all }\beta' <
        \beta.
      \]
    \item \label{item:livenessB} there are $t'$, $\eta'$ and
      $v \in \vinterpN{A+B}{\eta',\ol\emptyset,\omega,\beta}$ such that
      \[
        \state*{t}{\eta}{p} \forwardf{v} \state*{t'}{\eta'}{3-p} \text{ and }
        \state*{t'}{\eta'}{3-p} \in F(A,B,\omega,\beta') \text{ for all }\beta' <
        \beta.
      \]
    \end{enumerate}
\end{lemma}

The liveness result is now an easy consequence of the above lemma and
the fundamental property:
\begin{proof}[Proof of \autoref{thr:fairness}]
  By Theorem~\ref{cor:lrl} $\out\,(\unbox\,t) \in F(A,B,\omega,\omega+n)$
  for any $n$. Using \autoref{lem:fairness}, we can then show that we
  can extend any finite reduction sequence
  \[
    \state*{\out\,(\unbox\,t)}{\eta_0}{1} \forwardf{\interm_{p_1} v_1} \state*{t_1}{\eta_i}{p_1}
    \forwardf{\interm_{p_2}v_2} \state*{t_2}{\eta_0}{p_2} \forwardf{\interm_{p_3}v_3}
    \dots
    \forwardf{\interm_{p_n}v_n} \state*{t_n}{\eta_n}{p_n}
  \]
  with
  $\state*{t_n}{\eta_n}{p_n} \forwardf{\interm_{p_{n+1}}v_{n+1}}
  \state*{t_{n+1}}{\eta_{n+1}}{p_{n+1}}$ so that
  $\state*{t_{n+1}}{\eta_{n+1}}{p_{n+1}} \in
  F(A,B,\omega,\omega)$. Since $\forwardf{}$ is deterministic, this
  defines the desired infinite sequence of reductions. Moreover, since
  $\state*{t_{i}}{\eta_i}{p_{i}} \in F(A,B,\omega,\omega)$ for each
  $i$, and the index $\alpha$ decreases for every step of the form
  \ref{item:livenessA}, we know that only finitely many reduction
  steps after $\state*{t_{i}}{\eta_{i}}{p_{i}}$ are of the form
  \ref{item:livenessA}. Thus, there is a $j \ge i$ with
  $p_j \neq p_{i}$. In addition, given that $A$ and $B$ are value
  types, so is $A+B$, and we thus obtain by
  \autoref{lem:vinterp_value} that
  $\hastype{}{\interm_{p_i}\,v_i}{A+B}$ for all $i \ge 1$.
\end{proof}

\subsubsection{Causality}
\label{sec:causality}

We conclude by sketching the proofs for the corresponding operational
properties for the reactive step semantics. The proof idea is the same
but instead of setting heaps in $\ol\eta$ to the empty heap, we
construct $\ol\eta$ so that it contains the input stream.

Recall that after each step of the (reactive) step semantics, the
machine starts a new empty `later' heap. Each heap comes with its own
namespace. Let $l_i$ be the location that is picked as the first
location by the allocator after $i$ steps of the (reactive) step
semantics, i.e.\ $\allocate{\eta_i} = l_i$ where $\eta_i$ is the empty
heap after $i$ steps. Given a value $v$ and number $i\ge 0$, we write
$\eta^i_v$ for the heap $l_i \mapsto v :: l_{i+1}$. We further define
the following set of heap sequences:
\begin{align*}
  H(A,i) &= \setcom{\eta^i_{v_i};\eta^{i+1}_{v_{i+1}};\dots}{\forall j \ge
    i. \hastype{}{v_j}{A}}
\end{align*}
The locations $l_i$ are used to feed input to the reactive step
semantics. The following lemma shows that each $l_i$ has the right
semantic type:
\begin{lemma}
  \label{lem:placeholderVinterp}
  For all $i \ge 0$, $\hastype{}{v}{A}$, and $\ol\eta\in H(A,i+1)$,
  we have that
  $l_i \in \vinterpN{\Later (\Str A)}{\eta^i_v,\ol\eta,\omega,\beta}$.
\end{lemma}
The lemma is proved by a straightforward induction on $\beta$ using
Corollary~\ref{cor:lrl}.

We can now prove variants of Lemma~\ref{lem:productivity},
\ref{lem:termination}, and \ref{lem:fairness} for the reactive step
semantics. For the reactive step semantics of streams
$\forwards*{}{}$, we define the corresponding set of machine
configurations that are safe according to the logical relation as follows:
\[
  S^R(\ol\eta,B,\beta) = \setcom{\stateS{t}{\eta}{l_i}}{\ol\eta =
    \eta^i_v;\eta^{i+1}_w;\ol\eta'\land t\in \tinterpN{\Str B}{(\eta,\eta^i_v\tick\eta^{i+1}_w),\ol\eta',\omega,\beta}}
\]
This construction takes as additional parameter $\ol\eta$ drawn from
$H(A,i)$ which represents future input. We can then formulate the
corresponding productivity property as follows:
\begin{lemma}
  \label{lem:causality}
  Given $\stateS{t}{\eta}{l_i} \in S^R((\eta^i_v;\ol\eta),B,\beta+1)$
  and $\ol\eta\in H(A,i+1)$, then there are
  $\stateS{t'}{\eta'}{l_{i+1}} \in S^R(\ol\eta,B,\beta)$ and $v'$ such
  that
    \[
      \stateS{t}{\eta}{l_i} \forwards*{v}{v'} \stateS{t'}{\eta'}{l_{i+1}}.
    \]

    Moreover, if $B$ is a value type then $\hastype{}{v'}{B}$.
\end{lemma}
The constructions for $\forwardu*{}{}$ and $\forwardf*{}{}$ are
analogous and we can then prove the causality property of the reactive
step semantics.
\begin{proof}[Proof of \autoref{thr:causality}]
  We give the proof for part \ref{item:causalityStr} of the
  theorem. Part \ref{item:causalityUntil} and \ref{item:causalityFair}
  follow by a similar adaptation of the proofs of
  Theorems~\ref{thr:termination} and \ref{thr:fairness}, respectively.
  By Corollary~\ref{cor:lrl},
  $\unbox\,t \in \tinterpN{\Str A \to \Str
    B}{\eta^0_{v_1}\tick\eta^1_{v_2},\ol\eta,\omega,\beta}$ for all $\beta$ and for
  $\ol\eta = \eta^2_{v_3};\eta^3_{v_4};\dots$. By
  Lemma~\ref{lem:placeholderVinterp}
  $\adv\,l_0 \in \tinterpN{\Str A}{\eta^0_{v_1}\tick\eta^1_{v_2},\ol\eta,\omega,\beta}$
  and thus
  $\unbox\,t\, (\adv\,l_0) \in \tinterpN{\Str
    B}{\eta^0_{v_1}\tick\eta^1_{v_2},\ol\eta,\omega,\beta}$. Hence, we have
  $\stateS{\unbox\,t\, (\adv\,l_0)}{\emptyset}{l_0} \in
  U(\eta^0_{v_1};\eta^1_{v_2};\ol\eta,B,\beta)$ for any $\beta$. Similarly to the proof of
  \autoref{thr:productivity}, we can then use \autoref{lem:causality}
  to construct the desired infinite sequence of reduction steps.
\end{proof}

\section{Related work}
\label{sec:related:work}

%
%Arrowised FRP \reftodo is a different approach to the problem of providing of ensuring productivity in 
%FRP, which has been implemented in the Yampa library for Haskell. Rather than treating streams as
%primitive, Arrowised FRP gives access to these only through an arrow type of signal functions with
%a pre-specified collection of operations. This looses some of the simplicity of the original FRP and rules
%out types such as streams of streams. \rasmus{Check. Move to intro?}

The work by \citet{cave14fair} mentioned in the introduction defines a language with a modal operator $\Delay$
as well as inductive and coinductive types, but no guarded fixed points. They define a family of reduction
relations indexed by ordinals up to and including $\omega$. The relations corresponding to finite ordinals 
describe reductions up to finitely many steps, and the one at $\omega$ describes global behaviour. 
They give an interpretation of types as predicates on values indexed by ordinals up to and including $\omega$,
and similarly to our interpretation of types, the interpretation of $\Delay A$ at $\omega$ refers to the 
interpretation of $A$ also at $\omega$. Using this they prove strong normalisation, and sketch proofs of 
causality, productivity and liveness, but they do not prove lack of space leaks as done here. 
The motivation for omitting the guarded fixed point operator is exactly
the observation mentioned in the introduction that these equate inductive and coinductive types. Instead,
programming with coinductive types like streams must be done by coiteration. The present paper shows how
to refine the modal type system to combine the type system of LTL with the power of the fixed point operator, 
gaining simplicity in programming and productivity checking. The language of \citet{cave14fair} has more
general inductive and coinductive types than Lively RaTT (but not general guarded recursive types), see 
discussion in \autoref{sec:conclusion}. 
%
%The present paper shows how to combine these results with guarded 
%recursion. % Moreover, the always modality $\Box$ used here is different, since  
The idea of transfinite step indexing as used both here and by \citet{cave14fair}, 
has also been used to model countable non-determinism~\citep{bizjak2014model}
and distinguishing between logical and and concrete steps in program verification~\citep{svendsen2016transfinite}.
 
\citet{jeffrey2012} and \citet{Jeltsch2012} independently discovered the connection between FRP and
LTL. \citet{Jeltsch2012,jeltsch2013temporal} 
studied a category theoretic common notion of models of LTL and FRP. 
\citet{jeffrey2012} defined a language for FRP as an abstraction of a model defined in a functional 
programming language. Signals are defined directly as time-dependent values and LTL types are defined by
quantifying over time. While the native function space of the language contains all signal functions, a 
type of causal functions is definable in the language. In later work, \citet{jeffrey2014} extends modal FRP with
heterogeneous stream types, i.e., streams of elements whose types are given by a stream of types, and use 
this to encode past-time LTL. Unlike the present work, neither Jeltsch, nor Jeffrey define 
an operational semantics of programs, and therefore prove no operational metatheoretical results. 

To our knowledge, the first work to define a modal type theory for FRP with a guarded fixed point operator 
is that of \citet{krishnaswami2011ultrametric}. This line of work also studies type systems for 
eliminating implicit space and time leaks.
%, i.e., the problem of programs holding on to memory while continually allocating
%until they run out of space, 
%and implicit time leaks, i.e., the problem of programs becoming gradually slower.
\citet{krishnaswami2012higher} use linear types to statically bound the size of the dataflow graph 
generated by a reactive program, while \citet{krishnaswami13frp} defines a simpler type system, but
rules out space leaks using the techniques also used in the present paper.
%by evaluating programs on a machine with agressive garbage collection as in the present work. 
\citet{bahr2019simply} recast this work in the setting a Simply RaTT, which unlike \citet{krishnaswami13frp}
uses Fitch style for programming with modal types, and extend these results by identifying and 
eliminating a type of time leaks stemming from fixed points. 

The guarded fixed point operator was first suggested by \citet{nakano2000} and has since received
much attention in logics for program verification because it be used as a synthetic approach~\citep{ToT,appel2007very}
to step-indexing~\citep{appel01indexed}. 
Moreover, combining this with a notion of quantification over clocks~\citep{atkey2013productive} or a constant 
modality~\citep{clouston2016guarded} one can use guarded recursion
to encode coinduction. Guarded recursion forms part
of the foundation of the framework Iris~\citep{jung2015iris} for higher-order concurrent separation logic in Coq,
and a number of dependent type theories with guarded recursion have been 
defined~\citep{bahr2017clocks,BirkedalL:gdtt-conf,GCTT}. In the simply typed setting \citet{guatto2018} extends
this with a notion of time warps. The combination of guarded recursion and higher inductive types \citep{hottbook}
has also been used for modelling process calculi \citep{mogelbergPOPL2019,veltri2020formalizing}. 
Although related to the modal FRP calculi, these 
systems are usually much more expressive, since space and time leaks are ignored in their design. For example,
they all include an operation $A \to\Later A$ transporting data into the future, a known source of space leaks. 
%We expect that the ideas presented in this paper will be relevant also in the setting of the systems mentioned
%above, which all suffer from the problem that guarded fixed points make inductive and coinductive guarded
%types indistinguishable. 

\section{Conclusion and future work}
\label{sec:conclusion}

This paper shows how guarded fixed points can be
combined with liveness properties in modal FRP.
%
%The contributions of this paper are mainly conceptual: To illustrate how guarded fixed points can be
%combined with liveness properties in modal FRP. 
While properties such as termination, liveness and fairness are perhaps beyond the scope
of properties traditionally expressed in simply typed programming languages, they could naturally
occur as parts of program specifications in dependently typed languages
and proof assistants. 
%
%
%It is perhaps unclear if programmers in practice would
%want to capture liveness and fairness properties in a simply typed programming language. On the other 
%hand, in dependently typed languages and proof assistants where more precise properties of 
%programs can be stated and proved, LTL would be a very natural language for expressing properties of
%reactive programs, and liveness and fairness natural properties to prove. 
We therefore view
Lively RaTT as a conceptual stepping stone towards a dependently typed language for reactive programming. 

The results of this paper have been presented in the setting of functional reactive programming, but 
we expect that the ideas will be relevant also in the setting of guarded recursion as described in 
\autoref{sec:related:work}. In these settings, the fact that inductive and coinductive types coincide 
means that termination cannot be expressed directly. This leads to limitations in the setting of program
verification, e.g., when defining notions such as weak bisimulation for programs~\citep{paviotti2019} 
and processes. We expect that the tools developed here can be used in this respect once this work
has been adapted to guarded recursion and extended to dependent types.

%Future work also includes proving results on the lack of implicit space leaks in Lively RaTT. 
%Although Lively RaTT is based on Simply RaTT, the corresponding results proved for this
%by \citet{bahr2019simply} do not directly transfer because Lively RaTT allows multiple time ticks 
%in a context. Still we have chosen to keep the restrictions known to be necessary for eliminating space leaks 
%(e.g. disallowing unrestricted saving of data for later time steps)
%in Lively RaTT, with the hope of proving these properties in future work. 

Future work also includes extending Lively RaTT with general classes of inductive types. Note that
as seen here one must distinguish between ordinary inductive types such as $\Nat$ and temporal ones
such as the $\Until$ types, where the recursion involves time steps and the recursors 
therefore need to be stable. The temporal inductive types should be defined as a class of strictly positive
inductive types where the recursion variable appear under a $\Delay$, generalising the $\Until$ types 
used here. One could likewise add a class of coinductive types in the ordinary sense, but the 
temporal coinductive types are subsumed by the guarded recursive types. 

%
%\rasmus{Inductive types}

%%% Acknowledgments
%\begin{acks}                     
%This work was supported by a research grant (13156) from VILLUM FONDEN.
%\rasmus{Where does this go exactly? Leave out of anon. submission}
%       %% acks environment is optional
%                                        %% contents suppressed with 'anonymous'
%  %% Commands \grantsponsor{<sponsorID>}{<name>}{<url>} and
%  %% \grantnum[<url>]{<sponsorID>}{<number>} should be used to
%  %% acknowledge financial support and will be used by metadata
%  %% extraction tools.
%
%  % This material is based upon work supported by the
%  % \grantsponsor{GS100000001}{National Science
%  %   Foundation}{http://dx.doi.org/10.13039/100000001} under Grant
%  % No.~\grantnum{GS100000001}{nnnnnnn} and Grant
%  % No.~\grantnum{GS100000001}{mmmmmmm}.  Any opinions, findings, and
%  % conclusions or recommendations expressed in this material are those
%  % of the author and do not necessarily reflect the views of the
%  % National Science Foundation.
%\end{acks}

%% Bibliography
\bibliography{paper}


\begin{thebibliography}{37}

%%% the Web address.

\ifx \showCODEN    \undefined \def \showCODEN     #1{\unskip}     \fi
\ifx \showDOI      \undefined \def \showDOI       #1{#1}\fi
\ifx \showISBNx    \undefined \def \showISBNx     #1{\unskip}     \fi
\ifx \showISBNxiii \undefined \def \showISBNxiii  #1{\unskip}     \fi
\ifx \showISSN     \undefined \def \showISSN      #1{\unskip}     \fi
\ifx \showLCCN     \undefined \def \showLCCN      #1{\unskip}     \fi
\ifx \shownote     \undefined \def \shownote      #1{#1}          \fi
\ifx \showarticletitle \undefined \def \showarticletitle #1{#1}   \fi
\ifx \showURL      \undefined \def \showURL       {\relax}        \fi
% invisible to TeX
\providecommand\bibfield[2]{#2}
\providecommand\bibinfo[2]{#2}
\providecommand\natexlab[1]{#1}
\providecommand\showeprint[2][]{arXiv:#2}

\bibitem[\protect\citeauthoryear{Abel and Pientka}{Abel and Pientka}{2013}]%
        {Abel:Wellfounded}
\bibfield{author}{\bibinfo{person}{Andreas Abel} {and}
  \bibinfo{person}{Brigitte Pientka}.} \bibinfo{year}{2013}\natexlab{}.
\newblock \showarticletitle{Wellfounded Recursion with Copatterns: A Unified
  Approach to Termination and Productivity}. In
  \bibinfo{booktitle}{\emph{Proceedings ICFP 2013}}. \bibinfo{publisher}{ACM},
  \bibinfo{pages}{185--196}.
\newblock


\bibitem[\protect\citeauthoryear{Abel, Vezzosi, and Winterhalter}{Abel
  et~al\mbox{.}}{2017}]%
        {Abel:NBE:sized:types}
\bibfield{author}{\bibinfo{person}{Andreas Abel}, \bibinfo{person}{Andrea
  Vezzosi}, {and} \bibinfo{person}{Th{\'{e}}o Winterhalter}.}
  \bibinfo{year}{2017}\natexlab{}.
\newblock \showarticletitle{Normalization by evaluation for sized dependent
  types}.
\newblock \bibinfo{journal}{\emph{{PACMPL}}} \bibinfo{volume}{1},
  \bibinfo{number}{{ICFP}} (\bibinfo{year}{2017}),
  \bibinfo{pages}{33:1--33:30}.
\newblock
\urldef\tempurl%
\url{https://doi.org/10.1145/3110277}
\showDOI{\tempurl}


\bibitem[\protect\citeauthoryear{Appel and McAllester}{Appel and
  McAllester}{2001}]%
        {appel01indexed}
\bibfield{author}{\bibinfo{person}{Andrew~W. Appel} {and}
  \bibinfo{person}{David McAllester}.} \bibinfo{year}{2001}\natexlab{}.
\newblock \showarticletitle{An {Indexed} {Model} of {Recursive} {Types} for
  {Foundational} {Proof}-carrying {Code}}.
\newblock \bibinfo{journal}{\emph{ACM Trans. Program. Lang. Syst.}}
  \bibinfo{volume}{23}, \bibinfo{number}{5} (\bibinfo{date}{Sept.}
  \bibinfo{year}{2001}), \bibinfo{pages}{657--683}.
\newblock
\showISSN{0164-0925}
\urldef\tempurl%
\url{https://doi.org/10.1145/504709.504712}
\showDOI{\tempurl}
\newblock
\shownote{00283.}


\bibitem[\protect\citeauthoryear{Appel, Mellies, Richards, and Vouillon}{Appel
  et~al\mbox{.}}{2007}]%
        {appel2007very}
\bibfield{author}{\bibinfo{person}{Andrew~W Appel},
  \bibinfo{person}{Paul-Andr{\'e} Mellies}, \bibinfo{person}{Christopher~D
  Richards}, {and} \bibinfo{person}{J{\'e}r{\^o}me Vouillon}.}
  \bibinfo{year}{2007}\natexlab{}.
\newblock \showarticletitle{A very modal model of a modern, major, general type
  system}. In \bibinfo{booktitle}{\emph{Proceedings of the 34th annual ACM
  SIGPLAN-SIGACT symposium on Principles of programming languages}}.
  \bibinfo{pages}{109--122}.
\newblock


\bibitem[\protect\citeauthoryear{Atkey and McBride}{Atkey and McBride}{2013}]%
        {atkey2013productive}
\bibfield{author}{\bibinfo{person}{Robert Atkey} {and} \bibinfo{person}{Conor
  McBride}.} \bibinfo{year}{2013}\natexlab{}.
\newblock \showarticletitle{Productive coprogramming with guarded recursion}.
\newblock \bibinfo{journal}{\emph{ACM SIGPLAN Notices}} \bibinfo{volume}{48},
  \bibinfo{number}{9} (\bibinfo{year}{2013}), \bibinfo{pages}{197--208}.
\newblock


\bibitem[\protect\citeauthoryear{Bahr, Grathwohl, and M{\o}gelberg}{Bahr
  et~al\mbox{.}}{2017}]%
        {bahr2017clocks}
\bibfield{author}{\bibinfo{person}{Patrick Bahr}, \bibinfo{person}{Hans~Bugge
  Grathwohl}, {and} \bibinfo{person}{Rasmus~Ejlers M{\o}gelberg}.}
  \bibinfo{year}{2017}\natexlab{}.
\newblock \showarticletitle{The clocks are ticking: No more delays!}. In
  \bibinfo{booktitle}{\emph{32nd Annual {ACM/IEEE} Symposium on Logic in
  Computer Science, {LICS} 2017, Reykjavik, Iceland, June 20-23, 2017}}.
  \bibinfo{publisher}{{IEEE} Computer Society}, \bibinfo{address}{Washington,
  DC, USA}, \bibinfo{pages}{1--12}.
\newblock
\urldef\tempurl%
\url{https://doi.org/10.1109/LICS.2017.8005097}
\showDOI{\tempurl}


\bibitem[\protect\citeauthoryear{Bahr, Graulund, and M{\o}gelberg}{Bahr
  et~al\mbox{.}}{2019}]%
        {bahr2019simply}
\bibfield{author}{\bibinfo{person}{Patrick Bahr},
  \bibinfo{person}{Christian~Uldal Graulund}, {and}
  \bibinfo{person}{Rasmus~Ejlers M{\o}gelberg}.}
  \bibinfo{year}{2019}\natexlab{}.
\newblock \showarticletitle{Simply RaTT: a fitch-style modal calculus for
  reactive programming without space leaks}.
\newblock \bibinfo{journal}{\emph{Proceedings of the ACM on Programming
  Languages}} \bibinfo{volume}{3}, \bibinfo{number}{ICFP}
  (\bibinfo{year}{2019}), \bibinfo{pages}{1--27}.
\newblock


\bibitem[\protect\citeauthoryear{Birkedal, Bizjak, Clouston, Grathwohl,
  Spitters, and Vezzosi}{Birkedal et~al\mbox{.}}{2018}]%
        {GCTT}
\bibfield{author}{\bibinfo{person}{Lars Birkedal}, \bibinfo{person}{Ale{\v{s}}
  Bizjak}, \bibinfo{person}{Ranald Clouston}, \bibinfo{person}{Hans~Bugge
  Grathwohl}, \bibinfo{person}{Bas Spitters}, {and} \bibinfo{person}{Andrea
  Vezzosi}.} \bibinfo{year}{2018}\natexlab{}.
\newblock \showarticletitle{Guarded Cubical Type Theory}.
\newblock \bibinfo{journal}{\emph{Journal of Automated Reasoning}}
  (\bibinfo{date}{Jun} \bibinfo{year}{2018}).
\newblock


\bibitem[\protect\citeauthoryear{Birkedal, Grathwohl, Bizjak, and
  Clouston}{Birkedal et~al\mbox{.}}{2017}]%
        {clouston2016guarded}
\bibfield{author}{\bibinfo{person}{Lars Birkedal}, \bibinfo{person}{Hans~Bugge
  Grathwohl}, \bibinfo{person}{Ale{\v{s}} Bizjak}, {and}
  \bibinfo{person}{Ranald Clouston}.} \bibinfo{year}{2017}\natexlab{}.
\newblock \showarticletitle{The Guarded Lambda-Calculus: Programming and
  Reasoning with Guarded Recursion for Coinductive Types}.
\newblock \bibinfo{journal}{\emph{Logical Methods in Computer Science}}
  \bibinfo{volume}{12} (\bibinfo{year}{2017}).
\newblock


\bibitem[\protect\citeauthoryear{Birkedal, M{\o}gelberg, Schwinghammer, and
  St{\o}vring}{Birkedal et~al\mbox{.}}{2011}]%
        {ToT}
\bibfield{author}{\bibinfo{person}{Lars Birkedal},
  \bibinfo{person}{Rasmus~Ejlers M{\o}gelberg}, \bibinfo{person}{Jan
  Schwinghammer}, {and} \bibinfo{person}{Kristian St{\o}vring}.}
  \bibinfo{year}{2011}\natexlab{}.
\newblock \showarticletitle{First steps in synthetic guarded domain theory:
  Step-indexing in the topos of trees}. In \bibinfo{booktitle}{\emph{In Proc.
  of LICS}}. \bibinfo{publisher}{IEEE Computer Society},
  \bibinfo{address}{Washington, DC, USA}, \bibinfo{pages}{55--64}.
\newblock
\urldef\tempurl%
\url{https://doi.org/10.2168/LMCS-8(4:1)2012}
\showDOI{\tempurl}


\bibitem[\protect\citeauthoryear{Bizjak, Birkedal, and Miculan}{Bizjak
  et~al\mbox{.}}{2014}]%
        {bizjak2014model}
\bibfield{author}{\bibinfo{person}{Ale{\v{s}} Bizjak}, \bibinfo{person}{Lars
  Birkedal}, {and} \bibinfo{person}{Marino Miculan}.}
  \bibinfo{year}{2014}\natexlab{}.
\newblock \showarticletitle{A model of countable nondeterminism in guarded type
  theory}.
\newblock In \bibinfo{booktitle}{\emph{Rewriting and Typed Lambda Calculi}}.
  \bibinfo{publisher}{Springer}, \bibinfo{pages}{108--123}.
\newblock


\bibitem[\protect\citeauthoryear{Bizjak, Grathwohl, Clouston, M{\o}gelberg, and
  Birkedal}{Bizjak et~al\mbox{.}}{2016}]%
        {BirkedalL:gdtt-conf}
\bibfield{author}{\bibinfo{person}{Ale{\v{s}} Bizjak},
  \bibinfo{person}{Hans~Bugge Grathwohl}, \bibinfo{person}{Ranald Clouston},
  \bibinfo{person}{Rasmus~E M{\o}gelberg}, {and} \bibinfo{person}{Lars
  Birkedal}.} \bibinfo{year}{2016}\natexlab{}.
\newblock \showarticletitle{Guarded dependent type theory with coinductive
  types}. In \bibinfo{booktitle}{\emph{International Conference on Foundations
  of Software Science and Computation Structures}}. Springer,
  \bibinfo{pages}{20--35}.
\newblock


\bibitem[\protect\citeauthoryear{Cave, Ferreira, Panangaden, and Pientka}{Cave
  et~al\mbox{.}}{2014}]%
        {cave14fair}
\bibfield{author}{\bibinfo{person}{Andrew Cave}, \bibinfo{person}{Francisco
  Ferreira}, \bibinfo{person}{Prakash Panangaden}, {and}
  \bibinfo{person}{Brigitte Pientka}.} \bibinfo{year}{2014}\natexlab{}.
\newblock \showarticletitle{Fair {Reactive} {Programming}}. In
  \bibinfo{booktitle}{\emph{Proceedings of the 41st {ACM} {SIGPLAN}-{SIGACT}
  {Symposium} on {Principles} of {Programming} {Languages}}}
  \emph{(\bibinfo{series}{{POPL} '14})}. \bibinfo{publisher}{ACM},
  \bibinfo{address}{San Diego, California, USA}, \bibinfo{pages}{361--372}.
\newblock
\showISBNx{978-1-4503-2544-8}
\urldef\tempurl%
\url{https://doi.org/10.1145/2535838.2535881}
\showDOI{\tempurl}


\bibitem[\protect\citeauthoryear{Clouston}{Clouston}{2018}]%
        {Clouston:fitch-2018}
\bibfield{author}{\bibinfo{person}{Ranald Clouston}.}
  \bibinfo{year}{2018}\natexlab{}.
\newblock \showarticletitle{Fitch-style modal lambda calculi}. In
  \bibinfo{booktitle}{\emph{International Conference on Foundations of Software
  Science and Computation Structures}}. Springer, \bibinfo{pages}{258--275}.
\newblock


\bibitem[\protect\citeauthoryear{Clouston, Mannaa, M{\o}gelberg, Pitts, and
  Spitters}{Clouston et~al\mbox{.}}{2018}]%
        {clouston2018modal}
\bibfield{author}{\bibinfo{person}{Ranald Clouston}, \bibinfo{person}{Bassel
  Mannaa}, \bibinfo{person}{Rasmus~Ejlers M{\o}gelberg},
  \bibinfo{person}{Andrew~M. Pitts}, {and} \bibinfo{person}{Bas Spitters}.}
  \bibinfo{year}{2018}\natexlab{}.
\newblock \showarticletitle{Modal Dependent Type Theory and Dependent Right
  Adjoints}.
\newblock \bibinfo{journal}{\emph{CoRR}}  \bibinfo{volume}{abs/1804.05236}
  (\bibinfo{year}{2018}), \bibinfo{pages}{1--21}.
\newblock
\showeprint[arxiv]{1804.05236}
\urldef\tempurl%
\url{http://arxiv.org/abs/1804.05236}
\showURL{%
\tempurl}


\bibitem[\protect\citeauthoryear{Elliott and Hudak}{Elliott and Hudak}{1997}]%
        {FRAN}
\bibfield{author}{\bibinfo{person}{Conal Elliott} {and} \bibinfo{person}{Paul
  Hudak}.} \bibinfo{year}{1997}\natexlab{}.
\newblock \showarticletitle{Functional Reactive Animation}. In
  \bibinfo{booktitle}{\emph{Proceedings of the Second ACM SIGPLAN International
  Conference on Functional Programming}} \emph{(\bibinfo{series}{ICFP '97})}.
  \bibinfo{publisher}{ACM}, \bibinfo{address}{New York, NY, USA},
  \bibinfo{pages}{263--273}.
\newblock
\showISBNx{0-89791-918-1}
\urldef\tempurl%
\url{https://doi.org/10.1145/258948.258973}
\showDOI{\tempurl}


\bibitem[\protect\citeauthoryear{Fitch}{Fitch}{1952}]%
        {Fitch:Symbolic}
\bibfield{author}{\bibinfo{person}{Frederic~Benton Fitch}.}
  \bibinfo{year}{1952}\natexlab{}.
\newblock \bibinfo{booktitle}{\emph{Symbolic logic, an introduction}}.
\newblock \bibinfo{publisher}{Ronald Press Co.}, \bibinfo{address}{New York,
  NY, USA}.
\newblock


\bibitem[\protect\citeauthoryear{Guatto}{Guatto}{2018}]%
        {guatto2018}
\bibfield{author}{\bibinfo{person}{Adrien Guatto}.}
  \bibinfo{year}{2018}\natexlab{}.
\newblock \showarticletitle{A generalized modality for recursion}. In
  \bibinfo{booktitle}{\emph{Proceedings of the 33rd Annual ACM/IEEE Symposium
  on Logic in Computer Science}}. ACM, \bibinfo{pages}{482--491}.
\newblock


\bibitem[\protect\citeauthoryear{Hughes, Pareto, and Sabry}{Hughes
  et~al\mbox{.}}{1996}]%
        {HughesPS96}
\bibfield{author}{\bibinfo{person}{J. Hughes}, \bibinfo{person}{L. Pareto},
  {and} \bibinfo{person}{A. Sabry}.} \bibinfo{year}{1996}\natexlab{}.
\newblock \showarticletitle{Proving the Correctness of Reactive Systems Using
  Sized Types}. In \bibinfo{booktitle}{\emph{Conference Record of POPL'96: The
  23rd {ACM} {SIGPLAN-SIGACT} Symposium on Principles of Programming Languages,
  Papers Presented at the Symposium, St. Petersburg Beach, Florida, USA,
  January 21-24, 1996}}. \bibinfo{pages}{410--423}.
\newblock


\bibitem[\protect\citeauthoryear{Jeffrey}{Jeffrey}{2012}]%
        {jeffrey2012}
\bibfield{author}{\bibinfo{person}{Alan Jeffrey}.}
  \bibinfo{year}{2012}\natexlab{}.
\newblock \showarticletitle{{LTL} types {FRP:} linear-time temporal logic
  propositions as types, proofs as functional reactive programs}. In
  \bibinfo{booktitle}{\emph{Proceedings of the sixth workshop on Programming
  Languages meets Program Verification, {PLPV} 2012, Philadelphia, PA, USA,
  January 24, 2012}}, \bibfield{editor}{\bibinfo{person}{Koen Claessen} {and}
  \bibinfo{person}{Nikhil Swamy}} (Eds.). \bibinfo{publisher}{{ACM}},
  \bibinfo{address}{Philadelphia, PA, USA}, \bibinfo{pages}{49--60}.
\newblock
\showISBNx{978-1-4503-1125-0}
\urldef\tempurl%
\url{https://doi.org/10.1145/2103776.2103783}
\showDOI{\tempurl}


\bibitem[\protect\citeauthoryear{Jeffrey}{Jeffrey}{2014}]%
        {jeffrey2014}
\bibfield{author}{\bibinfo{person}{Alan Jeffrey}.}
  \bibinfo{year}{2014}\natexlab{}.
\newblock \showarticletitle{Functional Reactive Types}. In
  \bibinfo{booktitle}{\emph{Proceedings of the Joint Meeting of the
  Twenty-Third EACSL Annual Conference on Computer Science Logic (CSL) and the
  Twenty-Ninth Annual ACM/IEEE Symposium on Logic in Computer Science (LICS)}}
  \emph{(\bibinfo{series}{CSL-LICS '14})}. \bibinfo{publisher}{ACM},
  \bibinfo{address}{New York, NY, USA}, Article \bibinfo{articleno}{54},
  \bibinfo{numpages}{9}~pages.
\newblock
\showISBNx{978-1-4503-2886-9}
\urldef\tempurl%
\url{https://doi.org/10.1145/2603088.2603106}
\showDOI{\tempurl}


\bibitem[\protect\citeauthoryear{Jeltsch}{Jeltsch}{2012}]%
        {Jeltsch2012}
\bibfield{author}{\bibinfo{person}{Wolfgang Jeltsch}.}
  \bibinfo{year}{2012}\natexlab{}.
\newblock \showarticletitle{Towards a common categorical semantics for
  linear-time temporal logic and functional reactive programming}.
\newblock \bibinfo{journal}{\emph{Electronic Notes in Theoretical Computer
  Science}}  \bibinfo{volume}{286} (\bibinfo{year}{2012}),
  \bibinfo{pages}{229--242}.
\newblock


\bibitem[\protect\citeauthoryear{Jeltsch}{Jeltsch}{2013}]%
        {jeltsch2013temporal}
\bibfield{author}{\bibinfo{person}{Wolfgang Jeltsch}.}
  \bibinfo{year}{2013}\natexlab{}.
\newblock \showarticletitle{Temporal Logic with "Until", Functional Reactive
  Programming with Processes, and Concrete Process Categories}. In
  \bibinfo{booktitle}{\emph{Proceedings of the 7th Workshop on Programming
  Languages Meets Program Verification}} \emph{(\bibinfo{series}{PLPV '13})}.
  \bibinfo{publisher}{ACM}, \bibinfo{address}{New York, NY, USA},
  \bibinfo{pages}{69--78}.
\newblock
\showISBNx{978-1-4503-1860-0}
\urldef\tempurl%
\url{https://doi.org/10.1145/2428116.2428128}
\showDOI{\tempurl}


\bibitem[\protect\citeauthoryear{Jung, Swasey, Sieczkowski, Svendsen, Turon,
  Birkedal, and Dreyer}{Jung et~al\mbox{.}}{2015}]%
        {jung2015iris}
\bibfield{author}{\bibinfo{person}{Ralf Jung}, \bibinfo{person}{David Swasey},
  \bibinfo{person}{Filip Sieczkowski}, \bibinfo{person}{Kasper Svendsen},
  \bibinfo{person}{Aaron Turon}, \bibinfo{person}{Lars Birkedal}, {and}
  \bibinfo{person}{Derek Dreyer}.} \bibinfo{year}{2015}\natexlab{}.
\newblock \showarticletitle{Iris: Monoids and invariants as an orthogonal basis
  for concurrent reasoning}.
\newblock \bibinfo{journal}{\emph{ACM SIGPLAN Notices}} \bibinfo{volume}{50},
  \bibinfo{number}{1} (\bibinfo{year}{2015}), \bibinfo{pages}{637--650}.
\newblock


\bibitem[\protect\citeauthoryear{Krishnaswami}{Krishnaswami}{2013}]%
        {krishnaswami13frp}
\bibfield{author}{\bibinfo{person}{Neelakantan~R. Krishnaswami}.}
  \bibinfo{year}{2013}\natexlab{}.
\newblock \showarticletitle{Higher-order {Functional} {Reactive} {Programming}
  {Without} {Spacetime} {Leaks}}. In \bibinfo{booktitle}{\emph{Proceedings of
  the 18th {ACM} {SIGPLAN} {International} {Conference} on {Functional}
  {Programming}}} \emph{(\bibinfo{series}{{ICFP} '13})}.
  \bibinfo{publisher}{ACM}, \bibinfo{address}{Boston, Massachusetts, USA},
  \bibinfo{pages}{221--232}.
\newblock
\showISBNx{978-1-4503-2326-0}
\urldef\tempurl%
\url{https://doi.org/10.1145/2500365.2500588}
\showDOI{\tempurl}


\bibitem[\protect\citeauthoryear{{Krishnaswami} and {Benton}}{{Krishnaswami}
  and {Benton}}{2011}]%
        {krishnaswami2011ultrametric}
\bibfield{author}{\bibinfo{person}{Neelakantan~R. {Krishnaswami}} {and}
  \bibinfo{person}{Nick {Benton}}.} \bibinfo{year}{2011}\natexlab{}.
\newblock \showarticletitle{Ultrametric Semantics of Reactive Programs}. In
  \bibinfo{booktitle}{\emph{2011 IEEE 26th Annual Symposium on Logic in
  Computer Science}}. \bibinfo{publisher}{{IEEE} Computer Society},
  \bibinfo{address}{Washington, DC, USA}, \bibinfo{pages}{257--266}.
\newblock
\showISSN{1043-6871}
\urldef\tempurl%
\url{https://doi.org/10.1109/LICS.2011.38}
\showDOI{\tempurl}


\bibitem[\protect\citeauthoryear{Krishnaswami, Benton, and
  Hoffmann}{Krishnaswami et~al\mbox{.}}{2012}]%
        {krishnaswami2012higher}
\bibfield{author}{\bibinfo{person}{Neelakantan~R. Krishnaswami},
  \bibinfo{person}{Nick Benton}, {and} \bibinfo{person}{Jan Hoffmann}.}
  \bibinfo{year}{2012}\natexlab{}.
\newblock \showarticletitle{Higher-order functional reactive programming in
  bounded space}. In \bibinfo{booktitle}{\emph{Proceedings of the 39th {ACM}
  {SIGPLAN-SIGACT} Symposium on Principles of Programming Languages, {POPL}
  2012, Philadelphia, Pennsylvania, USA, January 22-28, 2012}},
  \bibfield{editor}{\bibinfo{person}{John Field} {and} \bibinfo{person}{Michael
  Hicks}} (Eds.). \bibinfo{publisher}{{ACM}}, \bibinfo{address}{Philadelphia,
  PA, USA}, \bibinfo{pages}{45--58}.
\newblock
\showISBNx{978-1-4503-1083-3}
\urldef\tempurl%
\url{https://doi.org/10.1145/2103656.2103665}
\showDOI{\tempurl}


\bibitem[\protect\citeauthoryear{MacLane and Moerdijk}{MacLane and
  Moerdijk}{2012}]%
        {maclane2012sheaves}
\bibfield{author}{\bibinfo{person}{Saunders MacLane} {and}
  \bibinfo{person}{Ieke Moerdijk}.} \bibinfo{year}{2012}\natexlab{}.
\newblock \bibinfo{booktitle}{\emph{Sheaves in geometry and logic: A first
  introduction to topos theory}}.
\newblock \bibinfo{publisher}{Springer Science \& Business Media}.
\newblock


\bibitem[\protect\citeauthoryear{M{\o}gelberg and Paviotti}{M{\o}gelberg and
  Paviotti}{2019}]%
        {paviotti2019}
\bibfield{author}{\bibinfo{person}{Rasmus~E M{\o}gelberg} {and}
  \bibinfo{person}{Marco Paviotti}.} \bibinfo{year}{2019}\natexlab{}.
\newblock \showarticletitle{Denotational semantics of recursive types in
  synthetic guarded domain theory}.
\newblock \bibinfo{journal}{\emph{Mathematical Structures in Computer Science}}
  \bibinfo{volume}{29}, \bibinfo{number}{3} (\bibinfo{year}{2019}),
  \bibinfo{pages}{465--510}.
\newblock


\bibitem[\protect\citeauthoryear{M{\o}gelberg and Veltri}{M{\o}gelberg and
  Veltri}{2019}]%
        {mogelbergPOPL2019}
\bibfield{author}{\bibinfo{person}{Rasmus~Ejlers M{\o}gelberg} {and}
  \bibinfo{person}{Niccol{\`o} Veltri}.} \bibinfo{year}{2019}\natexlab{}.
\newblock \showarticletitle{Bisimulation as path type for guarded recursive
  types}.
\newblock \bibinfo{journal}{\emph{Proceedings of the ACM on Programming
  Languages}} \bibinfo{volume}{3}, \bibinfo{number}{POPL}
  (\bibinfo{year}{2019}), \bibinfo{pages}{1--29}.
\newblock


\bibitem[\protect\citeauthoryear{Nakano}{Nakano}{2000}]%
        {nakano2000}
\bibfield{author}{\bibinfo{person}{Hiroshi Nakano}.}
  \bibinfo{year}{2000}\natexlab{}.
\newblock \showarticletitle{A modality for recursion}. In
  \bibinfo{booktitle}{\emph{Proceedings Fifteenth Annual IEEE Symposium on
  Logic in Computer Science (Cat. No.99CB36332)}}. \bibinfo{publisher}{{IEEE}
  Computer Society}, \bibinfo{address}{Washington, DC, USA},
  \bibinfo{pages}{255--266}.
\newblock
\showISSN{1043-6871}
\urldef\tempurl%
\url{https://doi.org/10.1109/LICS.2000.855774}
\showDOI{\tempurl}


\bibitem[\protect\citeauthoryear{Nilsson, Courtney, and Peterson}{Nilsson
  et~al\mbox{.}}{2002}]%
        {nilsson2002}
\bibfield{author}{\bibinfo{person}{Henrik Nilsson}, \bibinfo{person}{Antony
  Courtney}, {and} \bibinfo{person}{John Peterson}.}
  \bibinfo{year}{2002}\natexlab{}.
\newblock \showarticletitle{Functional Reactive Programming, Continued}. In
  \bibinfo{booktitle}{\emph{Proceedings of the 2002 ACM SIGPLAN Workshop on
  Haskell}} \emph{(\bibinfo{series}{Haskell '02})}. \bibinfo{publisher}{ACM},
  \bibinfo{address}{New York, NY, USA}, \bibinfo{pages}{51--64}.
\newblock
\showISBNx{1-58113-605-6}
\urldef\tempurl%
\url{https://doi.org/10.1145/581690.581695}
\showDOI{\tempurl}


\bibitem[\protect\citeauthoryear{Pnueli}{Pnueli}{1977}]%
        {LTL}
\bibfield{author}{\bibinfo{person}{Amir Pnueli}.}
  \bibinfo{year}{1977}\natexlab{}.
\newblock \showarticletitle{The temporal logic of programs}. In
  \bibinfo{booktitle}{\emph{18th Annual Symposium on Foundations of Computer
  Science (sfcs 1977)}}. IEEE, \bibinfo{pages}{46--57}.
\newblock


\bibitem[\protect\citeauthoryear{Sacchini}{Sacchini}{2013}]%
        {Sacchini13}
\bibfield{author}{\bibinfo{person}{Jorge~Luis Sacchini}.}
  \bibinfo{year}{2013}\natexlab{}.
\newblock \showarticletitle{Type-Based Productivity of Stream Definitions in
  the Calculus of Constructions}. In \bibinfo{booktitle}{\emph{28th Annual
  {ACM/IEEE} Symposium on Logic in Computer Science, {LICS} 2013, New Orleans,
  LA, USA, June 25-28, 2013}}. \bibinfo{pages}{233--242}.
\newblock


\bibitem[\protect\citeauthoryear{Svendsen, Sieczkowski, and Birkedal}{Svendsen
  et~al\mbox{.}}{2016}]%
        {svendsen2016transfinite}
\bibfield{author}{\bibinfo{person}{Kasper Svendsen}, \bibinfo{person}{Filip
  Sieczkowski}, {and} \bibinfo{person}{Lars Birkedal}.}
  \bibinfo{year}{2016}\natexlab{}.
\newblock \showarticletitle{Transfinite step-indexing: Decoupling concrete and
  logical steps}. In \bibinfo{booktitle}{\emph{European Symposium on
  Programming}}. Springer, \bibinfo{pages}{727--751}.
\newblock


\bibitem[\protect\citeauthoryear{{Univalent Foundations Program}}{{Univalent
  Foundations Program}}{2013}]%
        {hottbook}
\bibfield{author}{\bibinfo{person}{The {Univalent Foundations Program}}.}
  \bibinfo{year}{2013}\natexlab{}.
\newblock \bibinfo{booktitle}{\emph{Homotopy Type Theory: Univalent Foundations
  of Mathematics}}.
\newblock \bibinfo{publisher}{\url{https://homotopytypetheory.org/book}},
  \bibinfo{address}{Institute for Advanced Study}.
\newblock


\bibitem[\protect\citeauthoryear{Veltri and Vezzosi}{Veltri and
  Vezzosi}{2020}]%
        {veltri2020formalizing}
\bibfield{author}{\bibinfo{person}{Niccol{\`o} Veltri} {and}
  \bibinfo{person}{Andrea Vezzosi}.} \bibinfo{year}{2020}\natexlab{}.
\newblock \showarticletitle{Formalizing $\pi$-calculus in guarded cubical
  Agda}. In \bibinfo{booktitle}{\emph{Proceedings of the 9th ACM SIGPLAN
  International Conference on Certified Programs and Proofs}}.
  \bibinfo{pages}{270--283}.
\newblock


\end{thebibliography}

%% Appendix
% \appendix
% \section{Appendix}

% Text of appendix \ldots

\end{document}